\documentclass[10pt,conference,compsocconf,letterpaper]{IEEEtran}

\usepackage{cite}
\usepackage{ifthen}
\usepackage{epsfig, times, subfigure}
\usepackage{amssymb,amsmath,amsfonts}
\usepackage{graphicx}
\usepackage{makeidx,subfigure,verbatimfiles} % I used this for the perf eval lecture notes
\usepackage{array}
\usepackage[usenames]{color}
\usepackage{algorithm}
\usepackage{balance}
\usepackage{cases}
\usepackage{algorithmic}
\usepackage{verbatim}
\usepackage{changes}
\definechangesauthor{R}{red}
\definechangesauthor{A}{orange}
\definechangesauthor{K}{blue}
\definechangesauthor{G}{green}
\newtheorem{thm}{Theorem}
\newtheorem{lem}{Lemma}
\newtheorem{defi}{Definition}
\newtheorem{cor}{Corollary}

\newcommand{\fref}[1]{Fig.~\ref{#1}}
\newcommand{\sref}[1]{Section~\ref{#1}}
\newcommand{\cref}[1]{Chapter~\ref{#1}}

\begin{document}

\title{Greening File Distribution: Centralized or Distributed?}

\author{\IEEEauthorblockN{Kshitiz Verma\IEEEauthorrefmark{1}\IEEEauthorrefmark{2},
Gianluca Rizzo\IEEEauthorrefmark{1},
Antonio Fern\'andez Anta\IEEEauthorrefmark{1},
Rub\'en Cuevas Rum\'in\IEEEauthorrefmark{2},
Arturo Azcorra\IEEEauthorrefmark{2}}
\IEEEauthorblockA{\{kshitiz.verma, gianluca.rizzo, antonio.fernandez\}@imdea.org, \{rcuevas, aa\}@it.uc3m.es}
\IEEEauthorblockA{\IEEEauthorrefmark{2}Universidad Carlos III de Madrid, Spain}
\IEEEauthorblockA{\IEEEauthorrefmark{1}Institute IMDEA Networks, Madrid, Spain}
}

\maketitle

\begin{abstract}

Despite file-distribution applications are responsible for a major portion of the current Internet traffic, so far little effort has been dedicated to study file distribution from the point of view of energy efficiency. In this paper, we present a first approach at the problem of energy efficiency for file distribution. Specifically, we first demonstrate that the general problem of minimizing energy consumption in file distribution in heterogeneous settings is NP-hard. For homogeneous settings, we derive tight lower bounds on energy consumption, and we design a family of algorithms that achieve these bounds. Our results prove that collaborative p2p schemes achieve up to $50\%$ energy savings with respect to the best available centralized file distribution scheme. Through simulation, we demonstrate that in more realistic cases (e.g., considering network congestion, and link variability across hosts) we validate this observation, since our collaborative algorithms always achieve significant energy savings with respect to the power consumption of centralized file distribution systems.

%Recent studies have demonstrated the growing impact of the ICT sector in worldwide power consumption. Hence, a substantial amount of work has recently gone into defining novel energy efficient networking solutions.
%However, little effort has been dedicated to study file-distribution in this context, despite the fact that file-distribution applications are responsible for a major portion of the current Internet traffic.
%In this paper we present the first deep general study in the field of energy-efficiency for file-sharing. Specifically, we first demonstrate that the problem of minimizing energy consumption in file distribution is NP-Hard. Hence, we analyze restricted versions of the problem, yet maintaining a balance between simplicity and applicability in real networks. Our main contributions are: (i) we derive lower bounds on energy consumption, and design algorithms that achieve these bounds; (ii) we demonstrate that collaborative (p2p) schemes are more energy-efficient than centralized file-distribution systems. Finally, we conduct an empirical simulation study that, first, verifies the obtained analytical results and, second,  demonstrates that also in more realistic cases (e.g., considering network congestion) simple collaborative (or p2p-based) heuristics suffice to reduce the power consumption with respect to centralized file distribution systems.
\end{abstract}

%% A category with the (minimum) three required fields
%%A category including the fourth, optional field follows...
%\category{D.2.8}{Software Engineering}{Metrics}[complexity measures, performance measures]

%\terms{Theory, algorithms}

%\keywords{Energy efficiency, P2P file dissemination, theory, algorithms.}

%\section{Introduction}
%
%The need to compromise the satisfaction of energy demand with the control of global warming and carbon emissions 
%has triggered the interest on the design of novel energy-efficient solutions in multiple disciplines. 
%Specifically, recent studies reveal that the ICT (Information and Communications Technologies) sector is becoming 
%a major contributor to the worlwide energy consumption, overpassing even the aviation sector \cite{parliament}. 
%Furthermore, the energy consumption of the ICT sector is expected to double in the next decade \cite{pickavetEnergyConsumption}, unless new mechanisms and solutions are implemented. This situation has motivated the research community to investigate new mechanisms and solutions for saving energy, to be deployed by telecommunication network operators, Internet Service Providers (ISPs), content providers, and datacenter owners \cite{marcoMeliaEnergyCostISP,NikosDSLAMSigcom,ElasticTreesDataCenters,ValanciusNanoDataCenters}. 

\section{Introduction}

The need for a reduction in the carbon footprint of all human activities while satisfying an ever growing energy demand has triggered the interest on the design of novel energy-efficient solutions in several domains. Specifically, recent studies reveal that the ICT (Information and Communications Technologies) sector is becoming  a major contributor to the worldwide energy consumption, comparable to the aviation sector \cite{parliament}. 
Furthermore, the energy consumption of the ICT sector is expected to double in the next decade \cite{pickavetEnergyConsumption}, unless new mechanisms and solutions are implemented. This situation has motivated the research community to investigate novel mechanisms and solutions for saving energy in ICT, to be deployed by telecommunication network operators, Internet Service Providers (ISPs), content providers, and datacenter owners \cite{marcoMeliaEnergyCostISP,NikosDSLAMSigcom,ElasticTreesDataCenters,ValanciusNanoDataCenters}.
The proposed approaches in the field of energy efficient networking at either the device level (e.g. new hardware design \cite{bolla2011}) or the system level (energy efficient routing \cite{restrepo2009energy, andrews2010routing2} or sleep modes in wired and wireless networks \cite{gupta2003greening,agarwal2009somniloquy} aim to achieve an ``energy proportional'' network. This is, making the energy consumed by the network proportional to its traffic load. Specifically, hosts (servers and user terminals) are responsible of the major portion of the whole Internet power consumption \cite{pickavetEnergyConsumption}. Current energy efficient strategies in this domain aim at making the energy consumed proportional to the level of CPU or network activity of hosts, and often imply switching off or to a low power mode the devices when not active. However,  energy proportionality of hardware does not suffice to define a complete energy efficient framework for hosts. Indeed, new solutions must be found that implement energy efficient services (e.g. file sharing, web browsing, etc.) to optimize the utilization of hosts and network resources. 

In this paper, we focus on the file distribution service, which is one of the most widespread services on the Internet. Indeed, some of the existing file distribution services, such as peer-to-peer (p2p), one-click-hosting (OCH), software release, etc., represent a major fraction of current Internet traffic \cite{gkantsidis2006planet,LabovitzInternetTraffic,sandvineReportFall2011}. Despite of the importance of these services, to the best of the authors' knowledge, little effort has been dedicated to understanding and achieving energy-efficiency in the context of file distribution applications. In addition, within the context of corporate/LAN networks, other operations such as software updates are also file distribution processes. All this makes essential to deeply investigate energy-efficiency in file distribution, in order achieve a truly \emph{Green Internet}.

This paper is a first step into this direction. Our aim is to define the analytical and algorithmic basis for the design of energy efficient file distribution protocols. For this purpose, we first prove that the general problem of minimizing energy consumption in a file distribution process is NP-hard. Hence, we analytically study restricted versions of the problem, yet maintaining a balance between simplicity and applicability in real scenarios. Our analysis defines lower bounds and proposes collaborative p2p optimal (and near-optimal) algorithms for reducing energy consumption in the studied file distribution scenarios. Afterwards, we present an empirical evaluation through simulation, that allows us to $(i)$ validate our analytical results and $(ii)$ relax several assumptions imposed in the analytical study. %(due to the NP-hardness of the problem).
 Simulations show that, even in more realistic cases (considering energy costs associated to on-off state transitions or network congestion), our collaborative p2p schemes achieve significant energy savings with respect to centralized file distribution systems. These savings range between 50\% and two order of magnitude depending on the centralized scheme under consideration.

In summary, the main contributions of this paper are the following:

\begin{itemize}
\item We prove that the general problem of minimizing energy consumption in a file distribution process is NP-hard.
\item We derive lower bounds for the energy consumed in a file distribution process for simple yet realistic scenarios.
\item We design algorithms that achieve optimal (or near-optimal) energy consumption for these simple scenarios.
\item We demonstrate that the proposed collaborative p2p scheme is an appropriate approach to reduce the energy consumption in a file distribution process showing an improvement factor of at least $50\%$  with respect to any centralized file distribution schemes in the studied scenarios.

\item We perform an empirical simulation study that validates all the previous statements and quantify the energy savings achievable with our algorithms on a representative set of scenarios. 
\end{itemize}

The rest of the paper is structured as follows. Section \ref{sec:model} provides the network and energy model along with definitions and terminology used throughout the paper. Section \ref{sec:thanalysis} presents theoretical results obtained,
in the form of bounds and file distributions schemes.
In Section \ref{sec:numeval}, we present our simulation study. Section \ref{sec:relatedwork} revises the related work and Section \ref{sec:fwop} concludes the paper.

\section{System Model, Problem Definition and Assumptions}
\label{sec:model}

\subsection{System Model and Assumptions}
\label{subsec:netmodel}
%\added[G]{(Make it heterogeneous. Simplify notation)}

We consider a system of $n+1$ hosts ($n \geq 1$) that are fully connected via a wired network. 
%We assume full connectivity of the nodes, i.e., all the hosts are able to reach all the other hosts. 
One of these hosts, called the {\it server} and denoted by $S$, has initially a file of size $B$ that it has to distribute to all the other hosts, which we call the {\it clients}. 
We assume that the file is divided into $\beta \geq 1$ blocks of equal size $s=B/\beta$. 
%(Note that if $s=B,$ then $\beta=1$.)
The set of hosts is denoted as ${\cal H}=\{S,H_0,H_1, ..., H_{n-1} \}$, 
and the set of blocks as ${\cal B}=\{b_0, b_1, ..., b_{\beta-1}\}$.
We will also use in this paper a set of indexes, defined as ${\cal I}=\{S,0,\ldots,n-1 \}$. For simplicity of notation and presentation, 
we will often use an index $i \in {\cal I}$ to denote a host, and even talk about host $i$ instead of host $H_i$ (or $S$ when $i=S$).
 
% \par All the hosts can upload the blocks of the file to the other hosts (initially only $S$ can do so). A client can start uploading block $b_i$ only if it has received $b_i$ completely. Each host has upload capacity $u_i$ and download capacity $d_i$. Moreover, we assume that $d=ku$, for some integer $k \geq 1$. All the hosts are identical with respect to the processing speed and memory. No host can upload more than a block at any given time but can simultaneously upload and download from other host(s). Moreover, it can simultaneously download from up to $k$ hosts, assuming that all the links in the network have the sufficient capacity. 

All the hosts in ${\cal H}$ can potentially upload blocks of the file to other hosts (initially only $S$ can do so). A client can start uploading block $b_i$ only if it has received $b_i$ completely. 
Hosts have upload capacity $u_i$ and download capacity $d_i$, for $i \in {\cal I}$. (Observe that the server has upload capacity $u_S$.) We assume that all capacities are integral.
%\replaced[A][A: I do not think we use that $u \leq d$ anywhere]{We assume that all capacities are integral. }{We assume, as it is usual in practice, that all capacities are integral and that $\max_i\{u_i\}\leq\min_i\{d_i\}$.}
All the hosts are assumed to be identical with respect to processing speed, and to have enough memory to sustain the distribution process. 
No host can upload more than a block at any given time instant, but can simultaneously upload and download from other hosts. 
Moreover, it can simultaneously download from multiple hosts as long as the download capacity allows it. 
%We assume that all the links in the network have the sufficient capacity and the bottleneck is the host's capacity, not the link capacities of the network. 
We also assume that hosts always upload at their full capacity. 

We assume that time in the file distribution process is slotted. Each block transmission between hosts starts and finishes within the same slot.
We assume that no host uploads to more than one host in one slot. 
In general, the slot duration
may vary from one slot to the next. However, unless otherwise stated, we will assume during the rest of the paper that all slots have the same duration $\gamma$.
%
%Since each block size is $s$ and the upload capacity of each host $H_i$ is $u_i$, the duration of a time slot is given by $\frac{s}{\min\{u_i\}}$. 
Then, if the process of file distribution starts at time $t=0$, the time interval $[0, \gamma]$ corresponds to slot $\tau = 1$ and, in general, slot $\tau$ spans the time interval [$(\tau-1)\gamma, \tau \gamma$]. 
In each slot of a scheme, a host is assigned another host to serve (if any), and the set of blocks it will serve during that slot.
Note that hosts can only serve blocks that have been received completely.%(by the end of the previous slot).
%We assume for simplicity that the upload or download of the blocks start at the beginning of a slot.

In this work we consider only the energy consumed by hosts during the file distribution process. We do not consider the energy consumed by other network devices. In our model, the energy consumption has the following three components:
\begin{enumerate}
\item Each host $i \in {\cal I}$, just for being on, consumes power $P_i$  (when a host is off, we assume that it consumes no power). 
\item In addition, each host consumes $\delta_i \geq 0$, $i \in {\cal I}$ for each block served and/or received.
\item A host consumes energy while being switched on or off. If host $i \in {\cal I}$ takes time $\alpha_i$ to switch on or off, the energy consumed by switching is given by $P_i\alpha_i$.
\end{enumerate}

%
%\begin{table}[t]
%\caption{Some of the notation used in this work.}\label{tab:Notation}
%\begin{center}
%\begin{tabular*}{85.7mm}{|c|p{65mm}|}
%  \hline
%  \textbf{Symbol} & \textbf{Definition} \\
%  \hline\hline
%  $n$ & Total number of clients\\ \hline
%  $H_i$ & $i^{th}$ client \\ \hline
%  $S$ & Server, host that has the file initially  \\ \hline
%  ${\cal I}$ & Set of host indexes \\ \hline
%  $\beta$ & Number of blocks into which the file is divided\\ \hline
%  $b_j$ & $j^{th}$ block  \\ \hline
%  $B$ & Size of the file in bits \\ \hline
%  $s$ & Size of a block in bits \\ \hline
%  $u$ & Upload link speed (bits/s)\\ \hline
%  $d$ & Download link speed (bits/s)\\ \hline
%  $k$ & Ratio of the download to upload capacity ($d/u$) \\ \hline
%  $P_i$ & Power consumed by host $i$ when on (in Watt)\\ \hline
%   $\Delta_i$ & Energy consumed by host $i$ involved in a block transfer in a slot \\ \hline
%  %$\Delta_S$ & Energy consumed by the server in a block transfer in a slot \\ \hline
%  $\tau$ & Any arbitrary time slot \\ \hline
% % $\alpha$ & Energy consumed in switching on/off a host \\ \hline
% $z$ & A scheme to accomplish the distribution process \\ \hline
%  $c_{j,i}^z $ & Energy to transfer block $b_j$ to host $H_i$ under $z$ \\ \hline
%  $serv(j,i)$ & Index of the host that serves $b_j$ to host $H_i$ \\ \hline
%  $\mathcal{I}^z_{\tau}$ & Set of active hosts in slot $\tau$ under scheme $z$ \\ \hline
%  $\tau_f^z$ & Makespan of scheme $z$ \\ \hline
%\end{tabular*}
%\end{center}
%\normalsize
%\end{table}

\subsection{Problem and its Complexity}
% (talk about complexity)
% (talk about our approach) 

We define a \emph{file distribution scheme}, or \emph{scheme} for short, as a schedule of block transfers between hosts such that, after all the transfers, all the hosts have the whole file. Observe that a scheme must respect the model previously defined. Then,
the problem we study in this paper is defined as follows.
\begin{defi}
The \emph{file distribution energy minimization problem} is the problem of finding or designing a file distribution scheme that minimizes the total energy consumed.
\end{defi}
The bad news is that this problem is NP-hard even if switching on and off is free and there is no additional energy consumption per block (i.e., $\alpha_i=\delta_i = 0, \forall i \in {\cal I}$). Please refer to Appendix \ref{sec:np-hard} for the NP-hardness proof.
The good news is that, as will be shown later, even though the general problem is NP-hard, by making a few simplifying but still realistic assumptions, we can solve the file distribution energy minimization problem optimally. 

\subsection{Additional Assumptions}
\label{subsec:schemes}

Henceforth, we assume that all the hosts have the same upload capacity $u$, and the same download capacity $d$.
%, and the same additional energy consumption $\delta$ when transmitting or receiving.
We also assume that $\frac{d}{u}=k$ for some positive integer $k$. 
%\deleted[A][Does not seem to be useful]{Note that it is useless to have $k > \min\{n,\beta\}$, since one host uploads at least one block per slot and only to one host.
%Hence from now onwards we assume $1 \leq k \leq \min\{n,\beta\}$.}
Unless otherwise stated, we assume that hosts are switched on and off 
instantaneously, i.e., $\alpha_i=0, \forall i$, and hence switching consumes no energy.
%We assume that the switch off and switch on processes are instantaneous but consume energy $\alpha\geq0$. For analysis, we consider $\alpha=0$.

The uniformity of capacities results in a uniform slot duration, equal to $\gamma=\frac{s}{u}$, for all the block transfers. 
A host is said to be \emph{active} in a time slot if it is receiving or serving blocks in the slot. Otherwise, it is said to be \emph{idle}. 
The energy $\Delta_i$ consumed by an active host $i \in {\cal I}$ in one slot can be computed as follows.
\begin{equation}
\label{def:delta}
\Delta_i=P_i \gamma + \delta_i = \frac{P_is}{u}+\delta_i = \frac{P_iB}{u\beta}+\delta_i.
\end{equation}
Without loss of generality, we assume that %$P_0 \leq \cdots \leq P_{n-1}$, which implies that 
$\Delta_0 \leq \cdots \leq \Delta_{n-1}$. 

In some cases below we will assume that the system is \emph{energy-homogenous}. This means that all hosts have the same energy consumption parameters, i.e., $P_i = P$ and $\delta_i=\delta$, for all $i \in {\cal I}$. In such a homogeneous system, also all hosts have the same value of $\Delta_i=\Delta$. Note that, unless otherwise stated, we assume a heterogeneous system.

Let us consider parameters $n$, $k$, and $\beta$ of the file distribution energy minimization problem. Let us define the set of all possible schemes with these parameters by $\mathcal{Z}_k^{n,\beta}$. Let $E(z)$ be the energy consumed by scheme $z\in \mathcal{Z}_k^{n,\beta}$. 
\begin{defi}
%Among all possible schemes, 
A scheme $z_0 \in \mathcal{Z}_k^{n,\beta}$ is \emph{energy optimal} (or optimal for short) if $E(z_0) \leq E(z), \forall z \in \mathcal{Z}_k^{n,\beta}$.
\end{defi}
Hence, our objective in the rest of the paper is to find optimal (or quasi-optimal) schemes.

\subsection{Normal Schemes}

To rule out redundant and uninteresting schemes, we will consider only what we call normal schemes. 
Observe that the block transfers of a scheme $z$ in a slot $\tau$ can be modeled as a directed \emph{transfer graph} with the hosts as vertices and block transfers as edges (see \fref{f:transfer}). Then, a \emph{normal scheme} is a distribution scheme in which there are no idle hosts, there are no slots without active hosts, and each slot has a connected transfer graph. We denote the set of normal schemes with parameters $n$, $\beta$, and $k$ by $\mathcal{\hat Z}_k^{n,\beta}$. From now onwards, we will consider only normal schemes.
It is easy to observe that any optimal scheme can be transformed into a normal scheme that is also optimal. Hence, we are not
losing anything by concentrating only on normal ones.

Observe that in a transfer graph the out-degree of each vertex is at most 1 (by the upload constraint).
Thus, the transfer graph of a slot in a normal scheme can either be a tree (Fig. \ref{fig:graphtree}) 
or a graph with exactly one cycle (Fig. \ref{fig:graphcycle}).
Note also that in a slot with cycle all hosts upload blocks, while in a tree slot there are hosts that do not upload.

\begin{figure}[t]
\centering

\subfigure[Tree slot]
{
 \includegraphics[width=2.5cm]{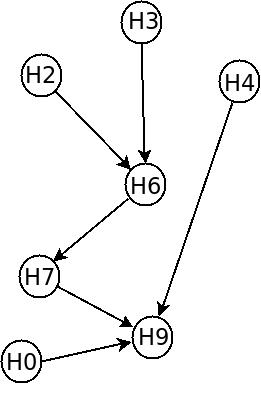}
 \label{fig:graphtree}
}
\qquad
\subfigure[Slot with cycle]
{
 \includegraphics[width=2.5cm]{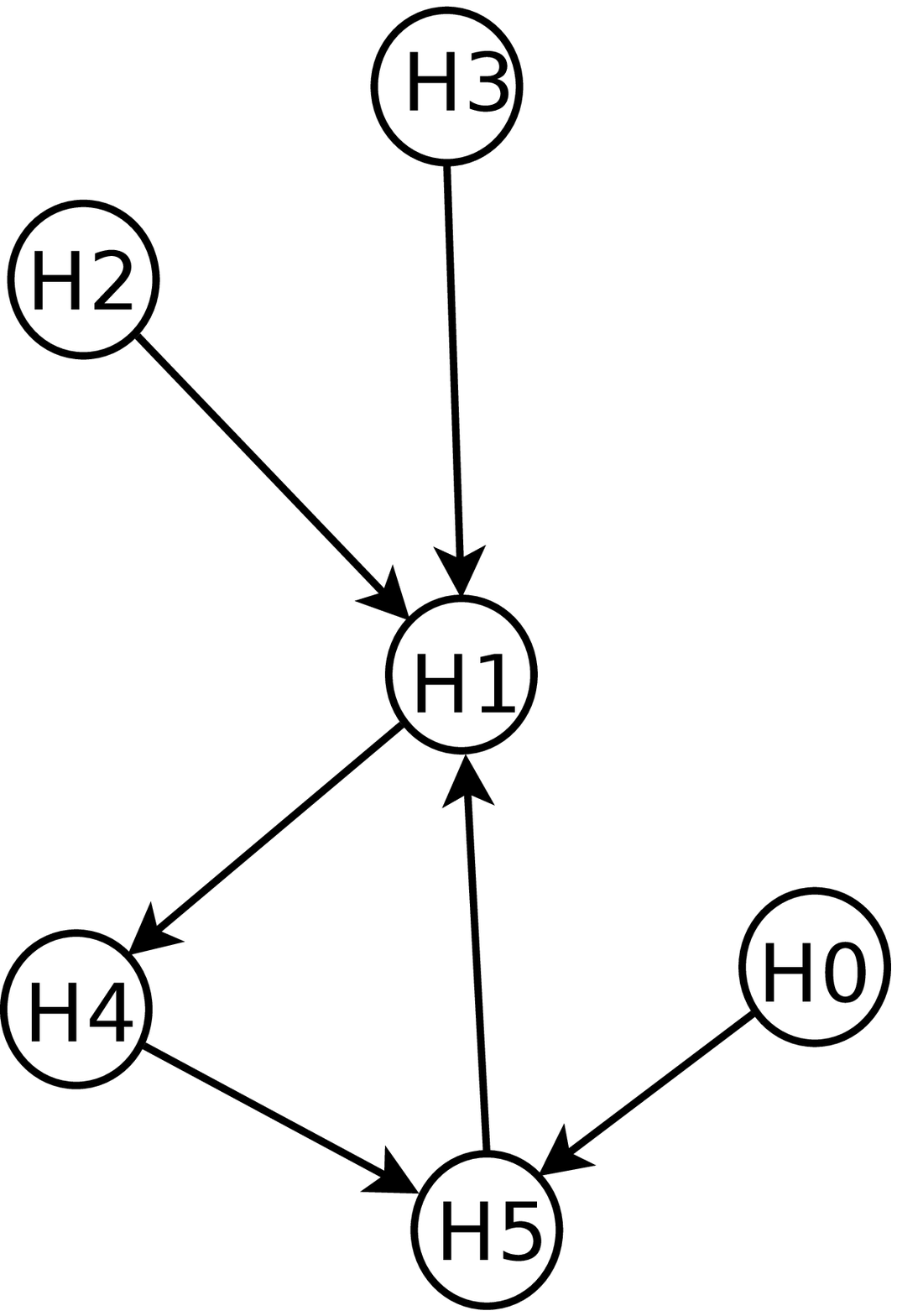}
 \label{fig:graphcycle}
}

\caption{A slot as a directed transfer graph. The number of blocks served in \ref{fig:graphcycle} is one more than the number of blocks served in \ref{fig:graphtree}, with the same energy consumption.}
\label{f:transfer}
\end{figure}
\vspace{-1mm}

\subsection{Costs}

%If define the set $\mathcal{\hat Z}_k$ of normal schemes when the download capacity is $k$ times the upload capacity. 
%Observe that $\mathcal{\hat Z}^{n,\beta}_1 \subseteq \mathcal{\hat Z}^{n,\beta}_2 \subseteq \cdots \subseteq \mathcal{\hat Z}^{n,\beta}_{\min\{n,\beta\}}$. 
%We denote the energy consumed by any scheme $z \in \mathcal{\hat Z}^{n,\beta}_k$ by $E_k(z)$, and the energy consumed by any optimal scheme $z_0 \in \mathcal{\hat Z}^{n,\beta}_k$ by $E_k^{\min}$. Let us denote the indexes of the set of active hosts in a time slot $\tau$ under a scheme $z$ by $\mathcal{I}_{\tau}^z$ (i.e., $\mathcal{I}_{\tau}^z \subseteq \{S,0,\ldots,n-1\}$). 
Let us consider scheme $z \in \mathcal{\hat Z}^{n,\beta}_k$. Denote with $\mathcal{I}_{\tau}^z \subseteq {\cal I}$ the indexes of the set of active hosts in time slot $\tau$ under scheme $z$. 
\begin{defi}
The \emph{cost of slot} $\tau$ under scheme $z$, denoted $c_{\tau}^z$, is the energy consumed by all active hosts $\mathcal{I}_{\tau}^z$ in $\tau$, i.e, 
\begin{equation}
c_{\tau}^z = \sum_{i\in\mathcal{I}_{\tau}^z}\Delta_i \nonumber
\end{equation}
\end{defi}

Let $\tau_f^z$ be the makespan of scheme $z$, i.e., the time slot of $z$ in which the distribution of the file is completed. Then, the energy consumed by the scheme $z$ can be obtained as 
\begin{equation}
\label{slotcost}
E(z) = \sum^{\tau_f^z}_{\tau=1}\sum_{i\in\mathcal{I}_{\tau}^z}\Delta_i
\end{equation}

The cost of a slot, as defined above, %considers the set of hosts that are active in a particular time slot,but it 
does not take into account which host is serving which block to which host. However, the total energy consumption 
of a scheme also depends on this. Thus, for a better insight on the schemes, we also associate a cost to a block transfer. 

We denote the set of blocks downloaded by host $i \in {\cal I}$ in slot $\tau$ under scheme $z$ by ${\mathcal{S}_{i,\tau}^z}$ and the index of the host serving $b_j \in {\mathcal{S}_{i,\tau}^z}$ as $serv(j,i)$.  

\begin{defi} We define the \emph{cost $c_{j,i}^z$ of a block} $b_j$ received by $H_i$ under scheme $z$ as,  
\label{costblockdef}
\begin{equation}
\label{costblock}
c_{j,i}^z = \mathcal{D}_{j,i}^z\cdot\Delta_i + \mathcal{U}_{j,i}^z\cdot \Delta_{serv(j,i)}
\end{equation}

where, if $b_j$ is received by $H_i$ in slot $\tau$,

 \[\mathcal{D}_{j,i}^z = \left\{ 
 \begin{array}{l l}
   1 & \quad \mbox{if $j=\min\{j'|b_{j'}\in \mathcal{S}_{i,\tau}^z\}$ }\\
   0 & \quad \mbox{Otherwise}\\ \end{array} \right. \]

\[\mathcal{U}_{j,i}^z = \left\{ 
\begin{array}{l l}
  1 & \quad \mbox{if $\mathcal{S}_{serv(j,i),\tau}^z=\emptyset$ }\\
  0 & \quad \mbox{Otherwise}\\ \end{array} \right. \]
\end{defi}

$\mathcal{D}_{j,i}^z$ accounts for the energy consumption of host $H_i$ (in units of $\Delta_i$) that is receiving the block. A block contributes to the energy consumed by $H_i$ if it is downloading. If a host is downloading more than one block in parallel, then we assume that only {\it one block} adds to the cost, as the rest of the blocks can be received without incurring any further cost. $\mathcal{U}_{j,i}^z$ accounts for the energy consumption of the host that is serving the block when $\mathcal{S}_{serv(j,i),\tau}^z=\emptyset$ (the host that is serving $b_j$ to $H_i$ is not downloading any block). 
%If a host is on for downloading a block, then we assume that it can serve a block without incurring any further cost (because it was on anyway). 
  
With the above definition, the sum of the costs of all blocks transferred in slot $\tau$ should be equal to the cost of the slot $\tau$, $c_{\tau}^z$. The next result establishes that this is indeed true for all the schemes. The proof can be found in Appendix~\ref{thmeq-proof}.

\begin{thm}
\label{thmeq}
The sum of the costs of all the blocks transferred during slot $\tau$ is equal to the cost of that slot, i.e.,
\begin{equation}
\label{equationeq}
\sum_{i \in \mathcal{I}_{\tau}^z } \sum_{b_j\in \mathcal{S}_{i,\tau}^z}  c_{j,i}^z = c_{\tau}^z
\end{equation}
\end{thm}

%\begin{proof}
%We transform the cost of a block as defined in Equation~\ref{costblock} to the following one. For each host $i \in \mathcal{I}_{\tau}^z$, define $\phi_i$ and $\psi_i$ as 

%\[\phi_{i} = \left\{ 
%\begin{array}{l l}
%  \Delta_i & \quad \mbox{if $\mathcal{S}_{i,\tau}^z\neq \emptyset$ }\\
%% 	& \quad \mbox{$w=$min$\{u|(u,v) \in E, d(u)=0\}$}\\
%  0 & \quad \mbox{Otherwise}\\ \end{array} \right. \]

%\[\psi_{i} = \left\{ 
%\begin{array}{l l}
%  \Delta_i & \quad \mbox{if $\mathcal{S}_{i,\tau}^z=\emptyset$ }\\
%  0 & \quad \mbox{Otherwise}\\ \end{array} \right. \]

%Note that $\sum_{b_j\in \mathcal{S}_{i,\tau}^z} \mathcal{D}_{j,i}^z =1$ iff $|\mathcal{S}_{i,\tau}^z|\geq 1$ (i.e., when $\phi_i = \Delta_i$). 
%It is easy to see that $\mathcal{U}_{j,i}^z =1$ iff $\psi_{serv(j,i)} = \Delta_{serv(j,i)}$, i.e., $S_{{serv(j,i)},\tau}^z=\emptyset$. 
%Therefore, for a host $i \in \mathcal{I}_{\tau}^z$, either $\phi_i=\Delta_i$ or $\psi_i=\Delta_i$, never both 0 or both $\Delta_i$. Hence,

%\[ \sum_{i\in \mathcal{I}_{\tau}^z} (\phi_i + \psi_i) = \sum_{i\in\mathcal{I}_{\tau}^z}\Delta_i \]
%\end{proof}

Thus, we can express the energy of a scheme $z$ in terms of the cost of blocks $c^z_{j,i}$ as 
% \begin{equation}
% \label{costschblock}
% E(z)= \sum_{i=1}^{n-1}\sum_{j=1}^{\beta}c_{j,i}^z
% \end{equation}

\begin{equation}
\label{costschblock}
E(z) = \sum_{i=0}^{n-1}\sum_{j=0}^{\beta-1}c_{j,i}^z =  \sum_{i=0}^{n-1}\sum_{j=0}^{\beta-1} \left(\Delta_i \cdot \mathcal{D}_{j,i}^z + \Delta_{serv(j,i)} \cdot \mathcal{U}_{j,i}^z\right)
\end{equation}
\section{Theoretical Analysis}
\label{sec:thanalysis}

In this section we provide analytical results for the file distribution energy minimization problem, under the additional assumptions described previously.
The results in this section are classified depending on the ratio $k$ between the download and upload capacities. First, we derive lower bounds on the energy consumption, and provide optimal schemes for the case $k=1$. For $k>1$, we provide optimal and near-optimal bounds and algorithms.

\subsection{Download Capacity = Upload Capacity}
\label{sec:homo}

%time taken to upload a block to a host is same as time taken to download a block in a slot. Thus, 
In this setting, a host can download at most one block during a slot. 
We first provide lower bounds on the energy consumed by any scheme.
Then, we present several optimal schemes, and we derive the value of $\beta$ that minimizes the energy of
optimal schemes in energy-homogenous systems.

%The number of blocks $(\beta)$ into which file is divided, is considered to be fixed and we may not change it unless stated otherwise.
%
%We find that there is no fundamental difference between the scenarios when all the hosts consume the same power in contrast to the scenario in which all the hosts consume different power. We also discuss the effect of the number of blocks on energy consumption of a scheme to distribute file in homogeneous power consumption scenario. To accomplish this, we relax the assumption that the value of $\beta$ is known in advance. 

\subsubsection{Lower Bound}
\label{subsec:bounds}

The following theorem provides a lower bound on the energy consumed by any distribution scheme when $k=1$.
\begin{thm} 
\label{minhetero}
The energy required by any scheme $z$ to distribute a file divided into $\beta$ blocks among $n$ clients when $k=d/u=1$, satisfies
$$
E(z) \geq \beta \left( \Delta_S + \sum_{i=0}^{n-1} \Delta_i \right) +  \max\{0,n-\beta\} \min\{\Delta_S,\Delta_0\}
$$
\end{thm}
The key observation behind this result
%behind the optimality when download capacity of a host is equal to its upload capacity 
is that each host has to be active for at least $\beta$ slots to receive the file, whereas the server has to be active for at least $\beta$ slots to upload one copy of each block among the clients. 
%One cannot do better than this which is the essence of the next theorem. 
The proof of the theorem can be found in Appendix~\ref{s-lb-du}.

\subsubsection{Optimal Distribution Schemes}
% and to provide some intuition behind them.\\
%Moreover, even though Algorithm \ref{algo1} is a special case of both Algorithms \ref{algo2} and \ref{algo3}, it is worth writing it separately because of its simplicity to convey the idea for optimality.\\
%design algorithms accordingly. 
%%Though dividing in two cases may be enough, dealing with the case $n = \beta$ separately makes it easy to understand the approach. 
%%We generalize it for $n>\beta$ and $n<\beta$ as well. 
%%
%The schemes are presented in the form of algorithms as Algorithms \ref{algo1}, \ref{algo2}, and \ref{algo3}. 

We now present optimal schemes achieving the lower bound of Theorem~\ref{minhetero}. 
%Note that Equation \ref{eq:minhetero} also hints that the optimal scheme may depend on the relation between $n$ and $\beta$. Hence, 
We distinguish among three cases, depending on the relation between $n$ and $\beta$, and we indicate the resulting schemes as Algorithms \ref{algo1}, \ref{algo2}, and \ref{algo3}. Note that in pseudocode, the transfer of block $b_j$ from host $H$ to host $H'$ is
expressed as $H \xrightarrow{{j}} H'$. Also, all the transfers that occur in the same slot are enclosed by the lines \emph{begin slot} and \emph{end slot}. While the three algorithms could be merged into a single one, we have chosen to present them separately for clarity.\\
We now provide some intuition on the algorithms. We start from Algorithm \ref{algo1}, which assumes that the number of clients is equal to the number of blocks. As each host has to be active at least $\beta$ slots to receive the complete file, Algorithm \ref{algo1} makes sure that the hosts are active for exactly $\beta$ slots. In the first $n$ slots of the algorithm, the server uploads a different block of the file to each of the $n$ clients. Since $n=\beta$, the server can upload the whole file to the clients in $n$ slots. Then the server goes off. At this point, all the hosts have one block and they all need to get the remaining $n-1$ blocks. Each client chooses a client to serve, in a way that the resulting transfer graph is a cycle of $n$ nodes. All the hosts start uploading the latest block they have received, and this process continues for $\beta-1$ slots, until all the hosts have all the blocks.\\
Algorithm \ref{algo2}, which assumes $n<\beta$, is more involved, but uses similar ideas as Algorithm \ref{algo1}. In Fig.~\ref{figtc}, we present a toy example of an scheme obtained from Algorithm~\ref{algo2}. In Algorithm \ref{algo3}, the number of clients is larger than the number of blocks. Thus some hosts will have to upload the same block more than once. In this algorithm, after that the server has served the first $\beta$ blocks, the host with the smallest energy consumption per slot uploads block $b_0$ to those hosts without any block.
\begin{thm}
\label{thmalgo}
When $d=u$, Algorithms \ref{algo1}, \ref{algo2}, \ref{algo3} describe optimal distribution schemes, with energy
$$
E(z) = \beta \left( \Delta_S + \sum_{i=0}^{n-1} \Delta_i \right) +  \max\{0,n-\beta\} \min\{\Delta_S,\Delta_0\}
$$
% and $P_i=P$, $\forall H_i\in \mathcal{H}$.
\end{thm}
For the proof, please refer to Appendix \ref{s-algo-proofs}. In what follows, with $Opt(n,\beta)$ we indicate the algorithm corresponding to the values of $n$ and $\beta$.
%
%We will refer to the schemes described by these algorithms as $Opt(n,\beta)$. The particular scheme/algorithm referred will depend on .
 
\begin{figure}[t!]
\centering
\subfigure[]
{
 \includegraphics[height=0.5in]{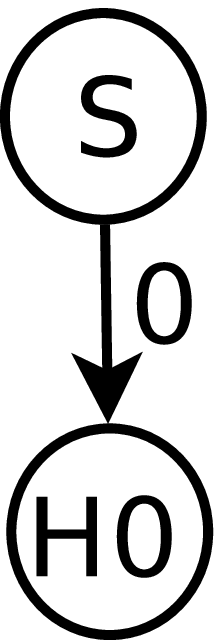}
 \label{fig:first_stepa}
}
~
\subfigure[]
{
 \includegraphics[height=0.5in]{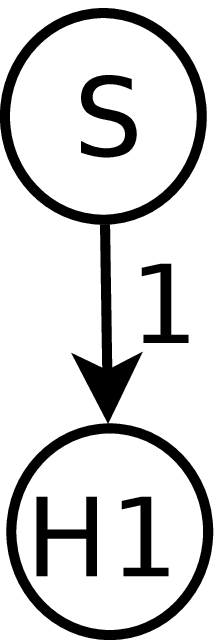}
 \label{fig:first_stepb}
}
~
\subfigure[]
{
 \includegraphics[height=0.5in]{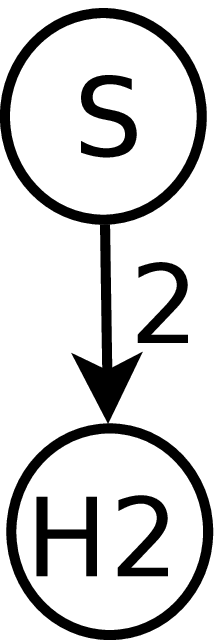}
 \label{fig:first_stepc}
}
~
\subfigure[]
{
\includegraphics[height=0.5in]{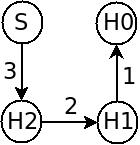}
\label{fig:second_step}
 }
~
\subfigure[]
{
  \includegraphics[height=0.5in]{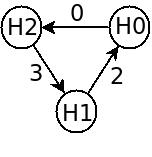}
  \label{fig:third_step}
}
~
\subfigure[]
{
\includegraphics[height=0.5in]{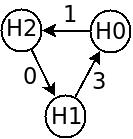}
\label{fig:fourth_step}
}
\caption{Example of Algorithm \ref{algo2}, for $n=3$ and $\beta=4$. The label on each arrow is the index of the block being served. 
%The first three transfers are performed one after the other such that all the hosts have a different block. After that, concurrent transfers are done. 
%The scheme has 4 tree slots and 2 slots with a cycle.
}
\label{figtc}
\end{figure}

\begin{algorithm}[th]
\caption{Optimal scheme for $\beta=n$}
\label{algo1}
\footnotesize
\begin{algorithmic}[1]
\FOR{$j=0:n-1$} \label{alg1:fl1}
\STATE {\it \textbf{begin slot}}
\STATE $S \xrightarrow{j} H_j$
\STATE {\it \textbf{end slot}}

\ENDFOR \label{alg1:fl1end}
\FOR{$j=n:2n-2$}\label{alg1:fl2}
\STATE {\it \textbf{begin slot}}
~~~~\FOR{$i=0:n-1$}
\STATE $H_i \xrightarrow{{(i+j)\bmod{n}}} H_{(i-1)\bmod{n}}$
\ENDFOR 
\STATE {\it \textbf{end slot}}
\ENDFOR \label{alg1:fl2end}
\end{algorithmic}
\end{algorithm}

\begin{algorithm}[th]
\caption{Optimal scheme for $\beta > n$}
\label{algo2}
\footnotesize
\begin{algorithmic}[1]
\FOR{$j=0:n-1$} \label{alg2:fl1}
\STATE {\it \textbf{begin slot}}
\STATE $S \xrightarrow{j} H_j$
\STATE {\it \textbf{end slot}}
\ENDFOR \label{alg2:fl1end}

\FOR{$j=n:\beta-1$ } \label{alg2:fl2} 
\STATE {\it \textbf{begin slot}}
\STATE $S \xrightarrow{j} H_{n-1} $ 
\FOR{$i=1:n-1$}
\STATE $H_i \xrightarrow{i+j-n} H_{i-1}$
\ENDFOR 
\STATE {\it \textbf{end slot}}
\ENDFOR \label{alg2:fl2end}

\FOR{$j=\beta:\beta+n-2$} \label{alg2:fl3}
\STATE {\it \textbf{begin slot}}
\FOR{$i=1:n$ }
\STATE $H_{i \bmod{n}} \xrightarrow{(i+j-n)\bmod{\beta}} H_{i-1}$
\ENDFOR
\STATE {\it \textbf{end slot}}
\ENDFOR \label{alg2:fl3end}
\end{algorithmic}
\end{algorithm}

\begin{algorithm}[th]
\caption{Optimal scheme for $\beta < n$. \newline 
$H_{\min}$ is the host with smallest $\Delta_i$. ($H_{\min} \in \{S,H_0\}$.) }
\label{algo3}
\footnotesize
\begin{algorithmic}[1]
\FOR{$j=0:\beta-1$} \label{alg3:fl1}
\STATE {\it \textbf{begin slot}}
\STATE $S \xrightarrow{j} H_j$
\STATE {\it \textbf{end slot}}
\ENDFOR \label{alg3:fl1end}

\FOR{$j=\beta:n-1$} \label{alg3:fl2}
\STATE {\it \textbf{begin slot}}
\STATE $H_{\min} \xrightarrow{0} H_{j+1-\beta}$ 
\FOR{$i=1:\beta-1$  } \label{alg3:flchange} 
\STATE $H_{i+j-\beta} \xrightarrow{i} H_{i+j+1-\beta}$
\ENDFOR                        \label{alg3:flchangend}
\STATE {\it \textbf{end slot}}
\ENDFOR \label{alg3:fl2end}

\FOR{$j=n:n+\beta-2$}  \label{alg3:fl3}
\STATE {\it \textbf{begin slot}}
\STATE $H_{2n-(j+1)} \xrightarrow{\beta-1} H_{n+\beta-(j+2)}$
\FOR{$i=0:\beta-2$}
\STATE $H_{(n+i-j)\bmod{n}} \xrightarrow{i} H_{(n+i-j-1)\bmod{n}}$
\ENDFOR
\STATE {\it \textbf{end slot}}
\ENDFOR \label{alg3:fl3end}
\end{algorithmic}
\end{algorithm}

%All our optimal schemes perform quite well with respect to the time taken to complete the distribution. They all finish the distribution in $\beta+n-1=O(n+\beta)$ slots. Thus, the algorithms presented here are not just energy optimal but time efficient as well. 

%-------------------------------------------

%\subsubsection{Energy dependence on number of blocks}
\subsubsection{Optimal Number of Blocks in Energy Homogenous Systems}
\label{subsec:tunable}

In this section we consider an energy-homogenous system, in which all hosts have the same energy consumption parameters, i.e., $P_i = P$ and $\delta_i=\delta$, for all $i \in {\cal I}$.
In this system we want to find the optimal value of $\beta$ into which the file should be divided for minimum energy consumption. Intuitively, the number of blocks into which the file must be divided depends on the value of $\delta$. If $\delta$ is very large, then it is better to divide the file in 
a small number of blocks, since each block transmission consumes additional energy $\delta$.
% because this will result in more number of transmissions, increasing the total energy consumed in transmissions. 
On the other hand, if $\delta$ is small, we can divide the file into a number of blocks such that the energy consumed is reduced due to
concurrent transfers.
% $\delta=0$ is a special case in which the energy consumed by a scheme does not depend on the number of blocks.

%
%\begin{equation}
%\label{costschblock}
%E(z)  =  \left( \frac{PB}{u\beta} + \delta \right) \sum_{i=0}^{n-1}\sum_{j=0}^{\beta-1} \left( \mathcal{D}_{j,i}^z + \mathcal{U}_{j,i}^z \right)
%\end{equation}

%Until now we had assumed that the number of blocks in which the file should be divided is given to us. In this subsection, we study the effect of choosing $\beta$ on the amount of energy consumed. From Equation \ref{costschblock} it is clear that the energy consumption of a scheme is inversely proportional to the number of blocks. Thus, we have to find the optimal value of $\beta$ into which the file should be divided for the minimum energy consumption. Intuitively, the number of blocks into which the file must be divided depends on the value of $\delta$. If $\delta$ is very large, then it is better not to divide the file in blocks because this will result in more number of transmissions, increasing the total energy consumed in transmissions. On the other hand, if $\delta$ is small, we can divide the file into a number of blocks such that the energy consumed is minimum.
%% $\delta=0$ is a special case in which the energy consumed by a scheme does not depend on the number of blocks.  

The following theorem presents the optimal value of $\beta$.
\begin{thm}
\label{thm-beta}
In a energy-homogenous system with $k=d/u=1$, the value of $\beta$ that minimizes the energy consumption of an optimal scheme is  
\begin{equation}
\label{valuenb-1}
\beta = \min\left\{\sqrt{\frac{PB}{u\delta}},n\right\}
\end{equation}
\end{thm}

%\begin{proof}
%From Theorems~\ref{minhetero} and \ref{thmalgo}, the energy consumption of an optimal scheme $z$ in a homogeneous system is
%\begin{equation}
%\label{mink=1}
%E(z) = \left( n\beta + \max\{n,\beta\} \right)\cdot\left( \frac{PB}{u\beta} + \delta \right)
%\end{equation}
%%
%To find the optimal value of $\beta$, we need to minimize the right hand side of Equation \ref{mink=1}. This can be written as a function of $\beta$ as
%%Using Equation \ref{mink=1}, we can rewrite Equation \ref{costschblock} as a function of $\beta$ and $\delta$ as,
%\begin{numcases}{E(\beta)=}
%\frac{PB}{u}(n+1) + \delta (n+1)\beta,&$ \beta \geq n$ \label{first} \\
%\frac{nPB}{u}\left(1+\frac{1}{\beta}\right)+\delta n (\beta+1),&$\beta \leq n$ \label{second}
%\end{numcases}
%%
%Note that in Equation \ref{first} the first term is a constant and the second is linear in $\beta$. This is a straight line with positive slope $\delta (n+1)$.
%Hence, the function attains the minimum at the lower extreme $\beta=n$, where it intersects Equation \ref{second}. Hence it is enough to consider
%Equation \ref{second} for $\beta \leq n$.
%Minimizing Equation \ref{second} with respect to $\beta$ we get, 
%\begin{equation}
%\label{valuenb-2}
%\beta = \sqrt{\frac{PB}{u\delta}}.
%\end{equation} 
%When this value is larger than $n$ the value $\beta=n$ has to be used.
%\end{proof}

Note that if the value of $\sqrt{\frac{PB}{u\delta}}$ is not an integer, it has to be rounded to one of the two closest integer values, such that $E(\beta)$ is minimum.

\subsection{Download Capacity $>$ Upload Capacity}
\label{subsec:d=ku}

In this subsection, we consider an energy homogenous system
%, which implies that all the hosts consume the same energy per slot $\Delta$, 
in which $k>1$. %In this scenario we provide lower and upper bounds on the energy consumed by an optimal scheme.

%a scenario in which all the hosts have the same power consumption $P$ which implies that the energy consumed per slot $\Delta_i=\Delta$ $\forall i\in\{0,1,...,n-1\}$ and the server. We derive lower and upper bounds on the minimum energy $E_k^{min}$, when the download capacity $(d)$ of the hosts is an integral multiple of its upload capacity $(u)$ for $2\leq \frac{d}{u}=k\leq \min\{n,\beta\}$. 

\subsubsection{Lower Bound}

In this section, we present a lower bound on the energy of a schedule in an energy homogenous system with $k>1$. %This bound is different from the one for $k=1$, because  
In this setting, the possibility to download more than one block in a slot implies that the minimum number of slots in which a host has to be on can be less than $\beta$.%, a host does not need to be on for for $k>1$ there is more flexibility in the schemes.
% as the value of $\mathcal{D}^z_{j,i}+\mathcal{U}^z_{j,i}$ depends on $k$. 
%It is because of this proving upper bound on the required number of tree slots as the number of blocks increase is hard. 
%The following theorem gives a lower bound on the energy consumed by an optimal schedule in a homogenous system. 
%Note that these bounds are derived assuming $k=2$ but are trivially true for any $k>2$ as well.    

\begin{thm}\label{th5}
Let $z$ be an optimal schedule in an energy homogenous system. Then the energy consumed by $z$ satisfies
\begin{equation}
\label{mink=2lb}
E(z) \geq n(\beta+1)\cdot\Delta
\end{equation}
\end{thm}
The derivation of this bound is based on proving that the required number of tree slots is at least $n$, because there are $n$ clients.
For the complete proof, please refer to Appendix \ref{sec:th5}.
%find the minimum number of tree slots. 
%Since a host can receive multiple blocks in a slot, the number of such slots may not be equal to $\max\{n,\beta\}$, as it is in the case of $d=u$. 
%prove that 
%
%
%The main difference between this case and the case $k=1$ is that in the latter case, there was no block with cost 0 (in fact $\mathcal{D}_{j,i}^z = 1$ always), whereas in the case $k>1$ there can be.

%The lower bound mentioned in Equation \ref{mink=2lb} is true for $d=ku, \forall k\geq1$. 

\subsubsection{(Quasi-)Optimal Distribution Schemes}

Observe that the energy consumption of Algorithms \ref{algo1} and \ref{algo3} in an energy homogenous system with $\beta \leq n$ is
exactly $n(\beta+1)\Delta$ (Theorem~\ref{thmalgo}). Hence, these algorithms describe optimal schemes for this system. 
%It is very interesting from this result that having a larger download does not help to reduce energy, 
%since these optimal schemes never use parallel downloads.
%
However, if $\beta > n$, the algorithm for $k=1$ (Algorithm \ref{algo2}) is not optimal anymore if $k>1$. 
In this section we present an algorithm, namely Algorithm \ref{algok=2}, that describes a distribution scheme for this case. 
In fact, the scheme works with $k=2$, as no host has more than two downloads in parallel.

Algorithm \ref{algok=2} distributes the file among the clients using ideas from Algorithms \ref{algo1} and \ref{algo2}. 
We represent the state of process with a two dimensional array $A$ of size $n \times \beta$ (Fig.~\ref{exp5}) with the rows and the columns representing 
the clients and the blocks, respectively. We set an entry $A_{ij}=1, i \in \{0,1,..,n-1\}, j\in \{0,1,..,\beta-1\}$ if and only if $H_i$ has received $b_j$, and 0 otherwise. 
At the beginning, all the entries are $0$ and after the completion of the algorithm they all should be $1$. 
Furthermore, imagine the array $A$ divided in $\lfloor \frac{\beta}{n} \rfloor - 1$
%$\frac{\beta-(n+b)}{n}$ 
square subarrays of size $n\times n$ and one rectangular subarray of size $n\times (n+b)$. (Note that this is just a conceptual division to understand Algorithm \ref{algok=2} in terms of Algorithms \ref{algo1} and \ref{algo2}.)

After the first loop, the diagonal of the first square subarray is set to 1, i.e., $A_{ii}=1, \forall i\in\{0,...,n-1\}$. Additionally, after the second loop,
the top left corner position (see Fig.~\ref{exp5}) of each subarray has also been set to 1, i.e., $A_{0j}=1$,$\forall j\in \{0,n,2n,..,(\lfloor \frac{\beta}{n} \rfloor - 1)n \}$. In each iteration of the for loop at Line \ref{shift}, the elements of one of the subarrays of $n\times n$ are set to 1 by serving in the same fashion as in Algorithm \ref{algo1}, while the server completes serving the diagonal of the next square/rectangular subarray. When Line \ref{run} is reached, all the elements of all the square subarrays are marked as 1. The remaining blocks are served using Lines \ref{alg2:fl2}-\ref{alg2:fl3end} of Algorithm \ref{algo2}, with an appropriate relabeling of the blocks.

\begin{algorithm}[th]
\caption{Energy saving scheme for case $k=2$ and $\beta > n$}
\label{algok=2}
\footnotesize
\begin{algorithmic}[1]
\STATE $b=\beta\bmod{n}$
\FOR{$j=0:n-1$} \label{alg5:fl1}
\STATE {\it \textbf{begin slot}}
\STATE $S \xrightarrow{j} H_j$ 
\STATE {\it \textbf{end slot}}
\ENDFOR \label{alg5:fl1end}

\FOR{$j=1: \lfloor \frac{\beta}{n} \rfloor - 1$} \label{alg5:fl2} 
\STATE {\it \textbf{begin slot}}
\STATE $S \xrightarrow{nj} H_0$
\STATE {\it \textbf{end slot}}
\ENDFOR \label{alg5:fl2end}

\FOR{$l=0:\lfloor \frac{\beta}{n} \rfloor - 2$} \label{shift} \label{alg5:fl3} 
\FOR{$j=0:n-2$}
\STATE {\it \textbf{begin slot}}
\STATE $S \xrightarrow{(l+1)n+j+1} H_{j+1}$ \label{algo5S} 
\FOR{$i=0:n-1$}
\STATE $H_i \xrightarrow{ln + ((i+j)\bmod{n})} H_{(i-1)\bmod{n}}$ \label{algo51}
\ENDFOR
\STATE {\it \textbf{end slot}}
\ENDFOR
\ENDFOR \label{alg5:fl3end}

\STATE Run Lines \ref{alg2:fl2}-\ref{alg2:fl3end} of $Opt(n,n+b)$ after renaming the block $b_{\beta-(n+b)+j}$ to $b_j,$ $\forall j\in \{0,1,..,n+b-1\}$ \label{run}
\end{algorithmic}
\end{algorithm}

\setlength{\unitlength}{1.5cm}
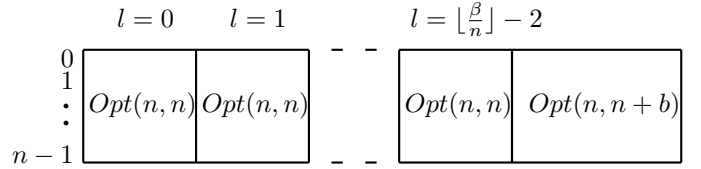
\begin{figure}
\vspace{-5mm}
\hspace{6.5mm}
\begin{centering}
\begin{picture}(1.8,1.8)

\thicklines
\put(0,0){\line(0,1){1}}
\put(0,0){\line(1,0){1}}
\put(0,1){\line(1,0){1}}
\put(1,0){\line(0,1){1}}
\put(1,0){\line(1,0){1}}
\put(1,1){\line(1,0){1}}
\put(2,0){\line(0,1){1}}
\put(2.2,1){\line(1,0){0.1}}
\put(2.5,1){\line(1,0){0.1}}
% \put(2.8,1){\line(1,0){0.1}}
\put(2.2,0){\line(1,0){0.1}}
\put(2.5,0){\line(1,0){0.1}}
% \put(2.8,0){\line(1,0){0.1}}
\put(2.8,0){\line(0,1){1}}
\put(2.8,0){\line(1,0){1}}
\put(2.8,1){\line(1,0){1}}
\put(3.8,0){\line(0,1){1}}
\put(3.8,0){\line(0,1){1}}
\put(3.8,0){\line(1,0){1.5}}
\put(3.8,1){\line(1,0){1.5}}
\put(5.3,0){\line(0,1){1}}
\put(0.3,1.2){$l=0$}
\put(1.3,1.2){$l=1$}
\put(2.9,1.2){$l=\lfloor \frac{\beta}{n} \rfloor - 2$}
\put(-0.2,0.85){$0$}
\put(-0.2,0.65){$1$}
\put(-0.15,0.5){\circle{0.01}}
\put(-0.15,0.35){\circle{0.01}}
\put(-0.63,0.0){$n-1$}
\put(0.05,0.45){$Opt(n,n)$}
\put(1.05,0.45){$Opt(n,n)$}
\put(2.85,0.45){$Opt(n,n)$}
\put(3.95,0.45){$Opt(n,n+b)$}
\end{picture}
\end{centering}
\caption{A representation of Algorithm \ref{algok=2} to visualize the distribution of blocks using the ideas of Algorithm \ref{algo1} and \ref{algo2}.}
\label{exp5}
\end{figure}

We present the bounds achieved in this section in the following
theorem. The proof of the second claim can be found in Appendix~\ref{sec:alg4}. 

\begin{thm}
\label{thm-alg4}
In a homogeneous system with $k>1$, 
\begin{itemize}
\item
If $\beta \leq n$, then Algorithms \ref{algo1} and \ref{algo3} describe optimal distribution schemes with energy $E(z)=n(\beta+1)\cdot\Delta$.
\item
If $\beta > n$, then Algorithm \ref{algok=2} describes a distribution scheme with energy 
\begin{equation}
\label{mink=2ub}
E(z)= \left( n(\beta+1) + \left\lfloor \frac{\beta}{n} \right\rfloor +b- 1 \right) \cdot \Delta
\end{equation}
where, $b=\beta\bmod{n}$
\end{itemize}
\end{thm}

While Algorithm \ref{algok=2} does not achieve optimal energy when $\beta > n$, it is quasi-optimal, since it is off from the lower bound by an additive term of 
$(\lfloor \beta/n \rfloor +b- 1) \Delta$, which is smaller than the term $n(\beta+1)\Delta$.
%
% If $\beta=O(n^p)$ for some constant $p\geq1$ then,
%\begin{equation}
%\label{eq:asymptotic}
%E_k^{min} = O(n^{p+1})+O(n^{p-1}) \nonumber
%\end{equation}
%which is asymptotically same as the lower bound on $E^{min}_k$.
%
It is important to note that Algorithm \ref{algok=2} uses $k=2$. 
%Since $\mathcal{\hat Z}_2 \subseteq \mathcal{\hat Z}_k$ for all $k \in \{2,...,{min\{n,\beta\}}\}$, 
Then, the upper bounds on the minimum energy presented here hold for all values of $k>1$. 
%This also suggests that the energy consumption does not depend on the ratio between download and upload capacities as long as all the hosts consume same power. 

%\input{np-hard}
% \input{forpaper}
\section{Performance Evaluation}
\label{sec:numeval}

In order to assess the performance of our scheme, we have run an extensive simulation study with two objectives. First, to evaluate quantitatively the results of our analysis in \sref{sec:thanalysis}. Second, to understand the impact on the performance of our schemes of some effects (like energy cost associated to on/off transitions, network congestion, or the variable power consumption among the devices involved in the file distribution process) not considered in our analysis, but typical of real scenarios.

%In order to quantify the amount of energy savings made possibe by our schemes, and to evaluate the impact of different factors which typically impact the energy consumption of file distribution in realistic scenarios, we have conduced an extensive set of simulations. 

\subsection{Experimental Setup}

In this section we briefly present a description of the experimental setup.

\subsubsection{Scenarios}
\label{sec-scenarios}
In our experiments we have considered two different scenarios, corresponding to two different application contexts for the file distribution problem.

\noindent \textbf{\textit{- Homogeneous scenario}}: In this case, all the hosts participating in the file distribution process have the same configuration. Specifically, we have considered the following values for the relevant input parameters in our experiments: nominal power $P= 80$ W, $\delta = 1$ Joule, and upload and download capacity $u=d=10$ Mbps. Finally, unless otherwise stated, we consider a scenario with one server and $200$ hosts.%, as well as \added[A][Do we really do this? The curves seems to be done with fixed block size]{the optimal block size defined in \sref{subsec:tunable}.}

This homogeneous scenario models a corporate network in which both the network infrastructure and the whole set of devices belong to the same company/organization, and are centrally managed. Typical file distribution processes in this context are software updates (e.g. OS, antivirus), which are usually centrally coordinated by system administrators. These environments are typically characterized by a relatively high uniformity in the network infrastructure and in the user terminals, especially if compared with the Internet. It is expected that communications among hosts in this type of intranet scenario happen at high bit rates, and that the bottleneck for file transfers happens at the terminals rather than in the network. Finally it is worth to mention that, in these settings, energy expenditure is a concern for the organization, as it directly impacts the OPEX of the IT infrastructure.

\noindent \textbf{\textit{- Heterogeneous scenario}}: In this setting, we analyze the impact of heterogeneity in host configurations on the performance of our schemes. This scenario captures the case in which hosts are typical Internet nodes (including home users), and it is therefore characterized by a significant variability across hosts in both the energy consumption profile and the observed network performance (i.e. different access speed and congestion conditions). In this case, the file distribution process is represented by, for instance, a software being released\footnote{Other applications such as entertainment content (video, music) file distribution also fit into this scenario.} (e.g., a new Linux distribution). In this scenario, the incentive for saving energy comes from corporate and indvidual sensibility towards reducing the carbon footprint, since the potential economical benefits for a single host are usually negligeable.%generally marginal.

In this setting we assume $u_i=d_i, \forall i\in\cal{I}$. In order to simplify our study, in our experiments we consider separately the effect of heterogeneity in power consumption and the effect of varying network conditions. %(we rather consider separatel homo and het scenarios). %for the hosts involved in the file distribution process. 
%We then refer to these scenarios as heterogeneous settings.

\subsubsection{File Distribution Schemes}

The file distribution schemes that we have considered in the performance evaluation are:

\noindent \textbf{\textit{- Opt}}: This is the file distribution scheme detailed in \sref{sec:homo}. It is a distributed scheme, since the upload capacity for distributing the file is made available by the same hosts that are downloading the file.

\noindent \textbf{\textit{- Parallel}}: This is a centralized scheme, in which all users download the same file at the same time from the same server in parallel. This is one of the most common architectures for file distribution, and it models a large number of file distribution services present in the current Internet (e.g., One Click Hosting systems such as Megaupload or RapidShare).

\noindent \textbf{\textit{- Serial}}: In this centralized scheme, the server uploads in sequence the complete file to the hosts involved in the file distribution process. That is, the server uploads the complete file to the first host. Once it finishes, it uploads the file to the second host, and so on. We consider this scheme because when $u_i=d_i$ it minimizes the amount of time each host is active in order to receive a file, and therefore the amount of energy spent by each host in the distribution process. This is realized at the expense of the server, who has to remain on for the whole duration of the scheme.% whose energy consumed in this scheme accounts for $50\%$ of the whole energy cost of the scheme (in the homogeneous scenario).    

\subsubsection{Energy Model}
\label{sec:emodel}

For our experiments we considered two different energy models. In a first one, the hosts only have two power states: an \textit{OFF} state, in which they do not consume anything, and an \textit{ON} state, in which they consume the full nominal power, equal to $80$W (typical nominal power consumption for notebooks and desktop PCs lies in the range $60$W-$80$W \cite{NordmanC09}). Unless otherwise stated, this is the default energy model for our experiments.\\
In order to understand the impact of load proportional energy consumption in our schemes, we consider a model that fits most of the current network devices \cite{NordmanC09}, in which the energy consumed has some dependency on the CPU utilization and network activity. This energy model is characterized by four states. Besides the OFF state, the other states are: the \textit{IDLE} state, in which the device is active but not performing any task, and consuming $80\%$ of the nominal power; the \textit{TX-or-RX} state, in which the device is active and either transmitting or receiving, and consuming $90\%$ of the nominal power; the \textit{TX-and-RX} state, in which the device is active and both transmitting and receiving, and consuming its full nominal power. We considered this model to analyze the impact of load proportionality on the overall energy consumption of the schemes considered in our experiments.\\
In \sref{sec-heter} we analyze the effect of having devices with heterogeneous power consumption profiles. For this purpose we use the previously described two-state model, but we assume that for each host its nominal power consumption is drawn from two different distribution: $(i)$ a Gaussian distribution with an average of $80$ W and a standard deviation of $20$ W, and $(ii)$ an exponential distribution, with an average of $80$ W. 

Note that, despite large servers typically present a larger nominal power, in our experiments we assign to the server the same nominal power as a regular host. This assumption is consistent with our intention to be conservative in our study, since our schemes require the server to be active far less time than the serial and parallel schemes.

\subsubsection{Goodness Metric}

The goodness metric we have used in order to compare the energy consumption of different file distribution schemes is \textsl{energy per bit}, computed as the ratio of the total amount of energy consumed by the distribution process, divided by the sum of the sizes of all the files delivered in the scheme. 

\subsection{Homogeneous Scenario}
\label{sec:num-hom}

%We start by considering the homogenous scenario described. 
%\deleted[A][I think this is not correct, maybe left from a previous version]{Furthermore, in this specific case we considered a file size of $100MB$, and a block size of 256kB footnote(This is a typical block size in several p2p systems such as BitTorrent) 
%Ruben: [We said in the description of the homogeneous scenario (Section IV.A.1 that we use the optimal block size. However it seems that are using 256KB. We need to clarify this and write down what is appropriate in sectoin IV.A.11], 
%hence a total of $400$ blocks.
%}

\subsubsection{Validation of the Analysis}
In Fig.~\ref{fig:421} we have plotted the energy per bit consumed by the file distribution process as function of the size of the file, for the three different file distribution schemes considered. As we can see, our schemes perform consistently better than both serial and parallel schemes. In particular, by maximizing the amount of time in which hosts serve while being served, our schemes tend towards reducing by half the total energy cost of serving a block with respect to the serial scheme.  This performance improvement with respect to the serial scheme is due to the use of (p2p-like) distribution, and indeed it decreases as the file size (and the number of blocks into which it is split) decrease. With respect to the serial scheme, our optimal schemes make the most out of the energy consumed by all hosts which are active and being served at a given time, by having them contributing as much as possible to the file distribution. As a consequence, despite each host spends more time in an active state than in the serial scheme, the net effect is a decrease of the total energy.

Moreover, we can also observe how the parallel scheme performs consistently worse than any other scheme, consuming up to two orders of magnitude more than the serial scheme. Since the utilization of this parallel scheme is widespread in the current Internet, our observations confirm the great potential of distributed schemes for saving energy. 

Fig.~\ref{fig:421} also depicts the performance of our \emph{Opt} algorithm for different number of hosts ($50$, $200$, and $400$). We observe that the energy per bit consumed by our algorithm as well as by the serial scheme are not affected by the number of hosts in the scheme. Hence for the rest of the section we will present results exclusively for a setting with $200$ hosts.

Finally, it is worth noting that, for the optimal scheme, the nonsmooth variation of the energy per bit with file size, observable at low values of file size, is due to quantization in the number of blocks. The serial and parallel schemes (for which there is no partition of the file into blocks) have a smoother behavior with respect to file size. %We address this issue in the next subsection.

\subsubsection{Block Size}
The impact of the total number of blocks on the energy consumed by our \emph{Opt} scheme can be seen in Fig.~\ref{fig:422}, where we plotted the energy per bit consumed with \emph{Opt} for variable file sizes, and for a total of $200$ hosts. The green curve corresponds to the case in which a fixed block size, equal to $256$ kB, is used, while the lower red one is obtained by using an optimal block size, according to the formula in \sref{subsec:tunable}. We see how the use of an optimal block size leads to an increment in energy savings mainly for small file sizes. The reason is that for small file sizes a fixed block size leads to a small number of blocks, and consequently to exploit less the distributed (p2p-like) mechanisms which, in our scheme, improve the efficiency of the distribution process. %\added[G]{Put n of blocks in the plot, (2 y axes)}

%  Here too we see that when on/off costs are considered, the performance gains over the serial scheme decrease for some values of file size, and therefore of number of blocks. One possibility to recover part of those gains is to optimize the energy cost of the scheme over the total number of blocks. Indeed, similarly to what done in sec. \ref{subsec:tunable}, it can be shown that when on/off costs are present, the optimal value of $\beta$ is given by 
%\[\beta=\sqrt{\frac{nPB}{u(\delta n+4Pt)}}\]
%As shown also in figure \ref{fig:fsalt}, 
%As shown in the figure, the use of the optimal number of blocks allows to improve the performance gains of our scheme over the serial scheme. \fref{fig:422} also shows that for file sizes which are not much larger than the block size, the amount of energy savings of our scheme when on/off costs are present is reduced. This is expected, as in realistic settings the use of any scheme which involves using the hosts to serve other hosts implies some amount of overhead which inevitably translates into additional energy costs.\\
%

\begin{figure*}[t]
\begin{minipage}[b]{0.33\linewidth}\centering
    \includegraphics[width=\textwidth]{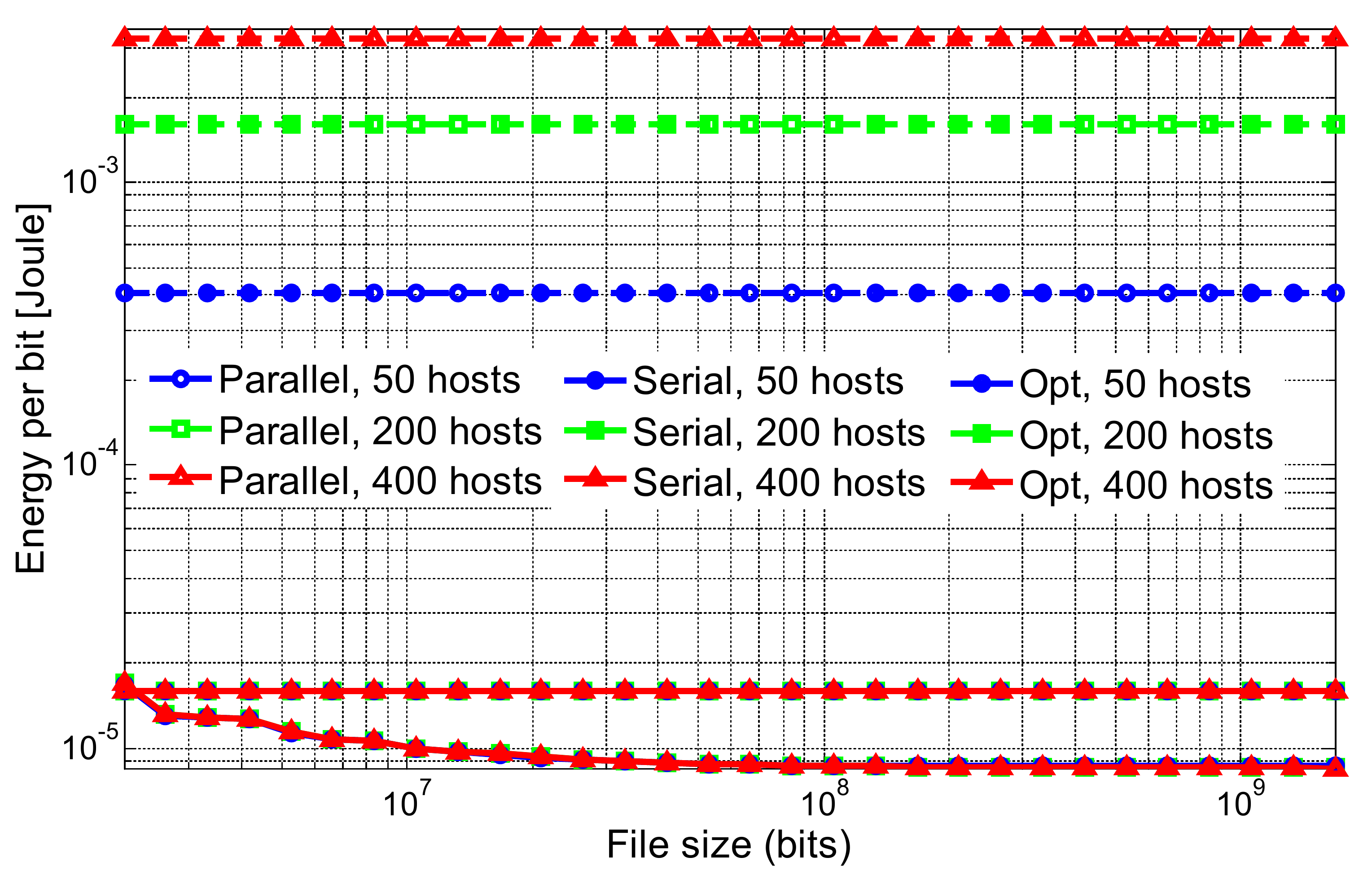}
 \vspace{-0.3in}
    \caption{Energy per bit consumed by \emph{Opt} in function of file size, compared with the serial and the parallel scheme. Block size: $256$kB.}\label{fig:421}
\end{minipage}
\begin{minipage}[b]{0.33\linewidth}\centering
    \includegraphics[width=\textwidth]{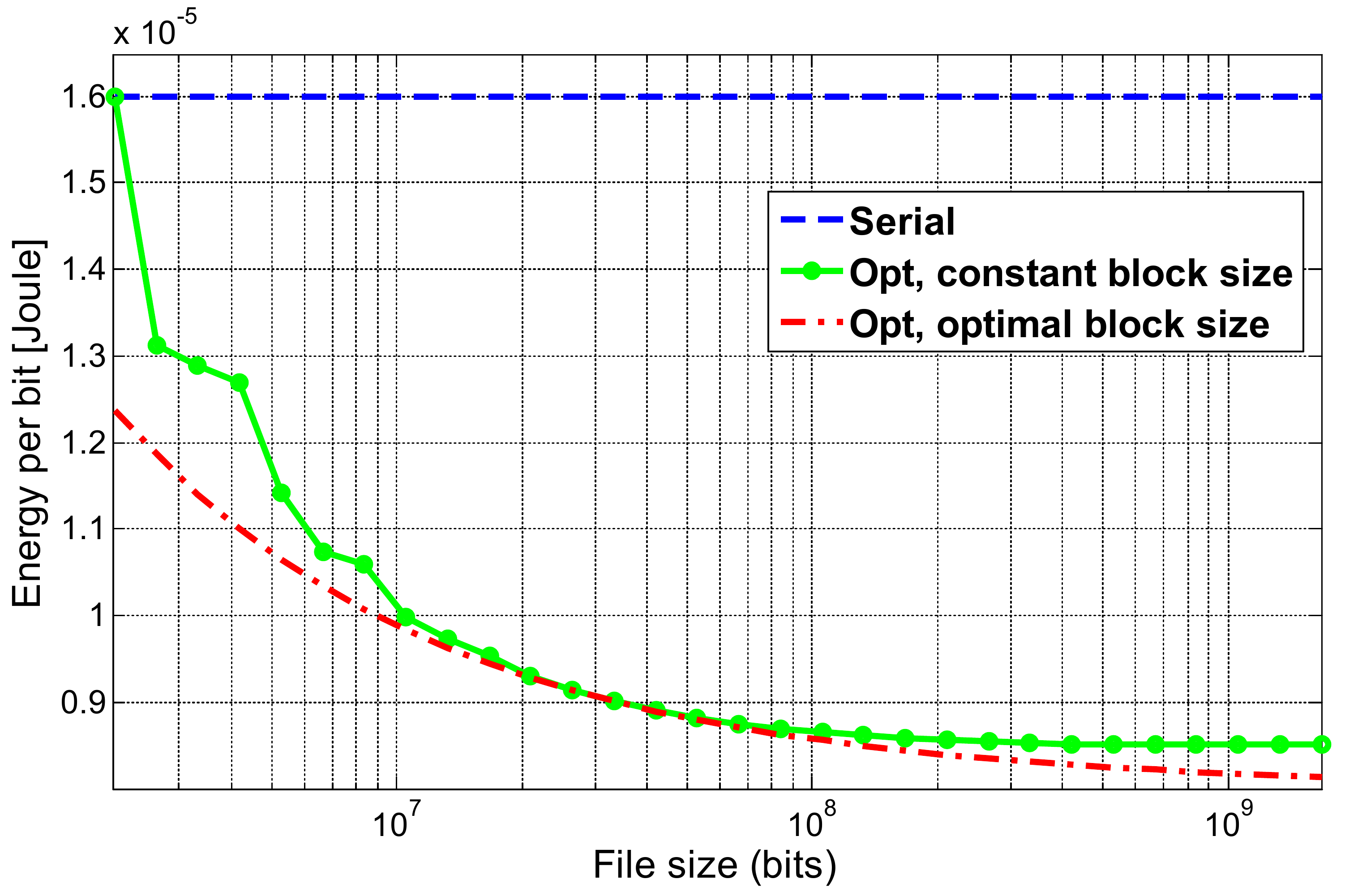}
  \vspace{-0.3in}
    \caption{Impact of the choice of number of blocks on the energy per bit consumed by our algorithm, in function of file size, with $200$ hosts.}\label{fig:422}
\end{minipage}
\begin{minipage}[b]{0.33\linewidth}\centering
    \includegraphics[width=\textwidth]{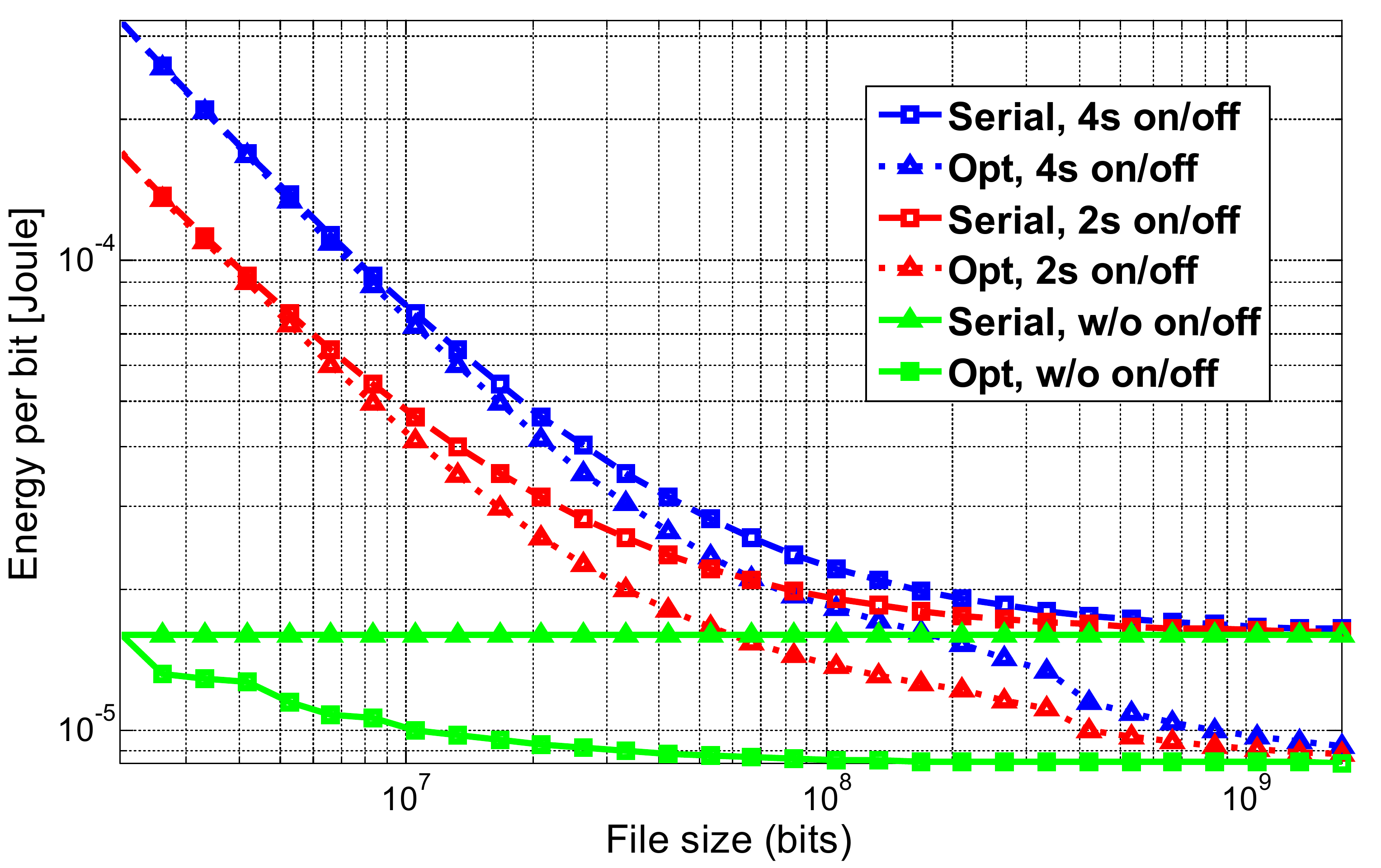}
  \vspace{-0.3in}
    \caption{Impact of on/off energy cost on the energy per bit consumed by our algorithm, in function of file size, with $200$ hosts.}\label{fig:423}
\end{minipage}

\begin{minipage}[b]{0.33\linewidth}\centering
    \includegraphics[width=\textwidth]{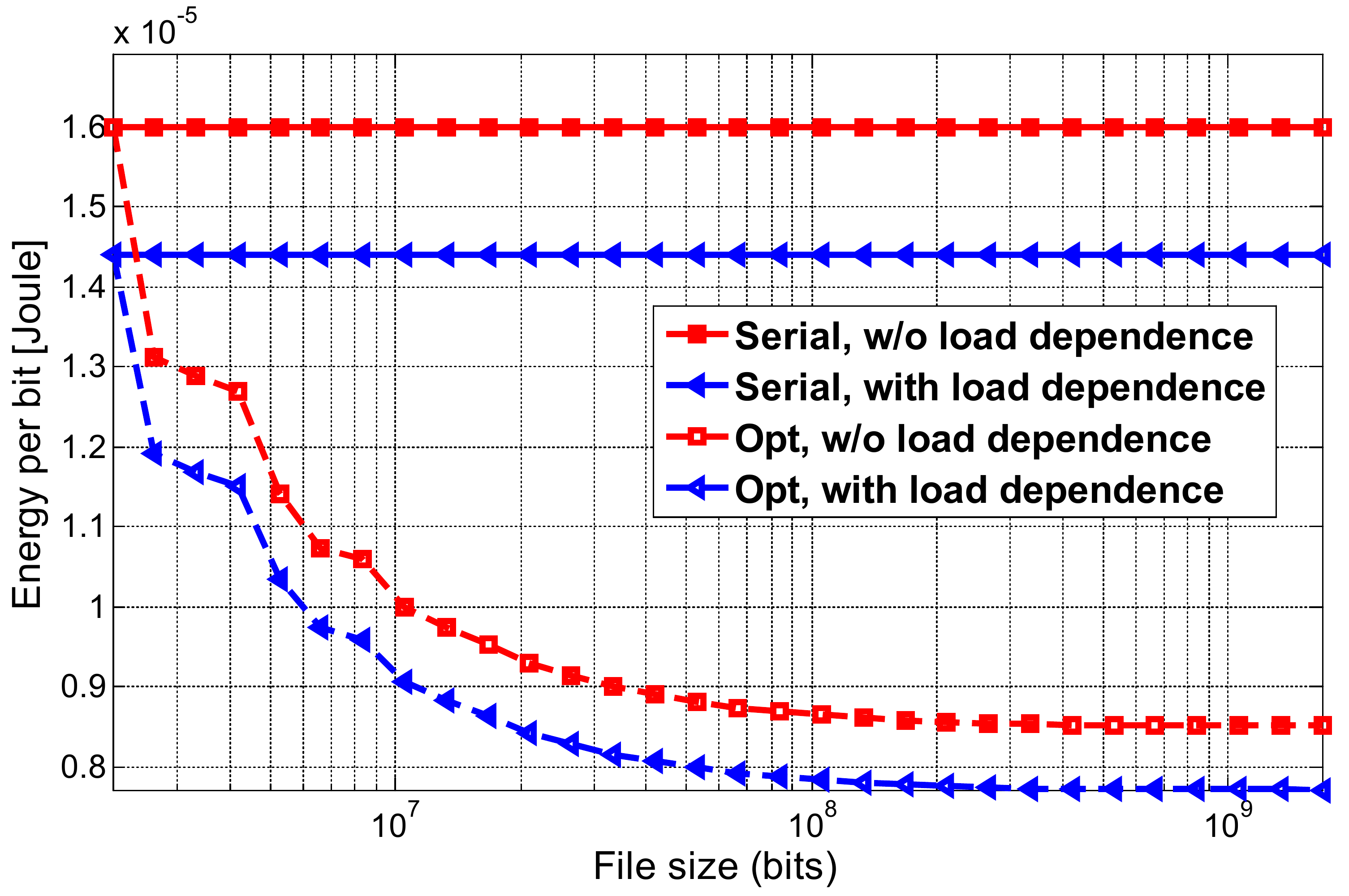}
  \vspace{-0.3in}
    \caption{Impact of the energy model on the energy per bit consumed by our algorithm, in function of file size, with $200$ hosts.}\label{fig:424}
\end{minipage}
\begin{minipage}[b]{0.33\linewidth}\centering
    \includegraphics[width=\textwidth]{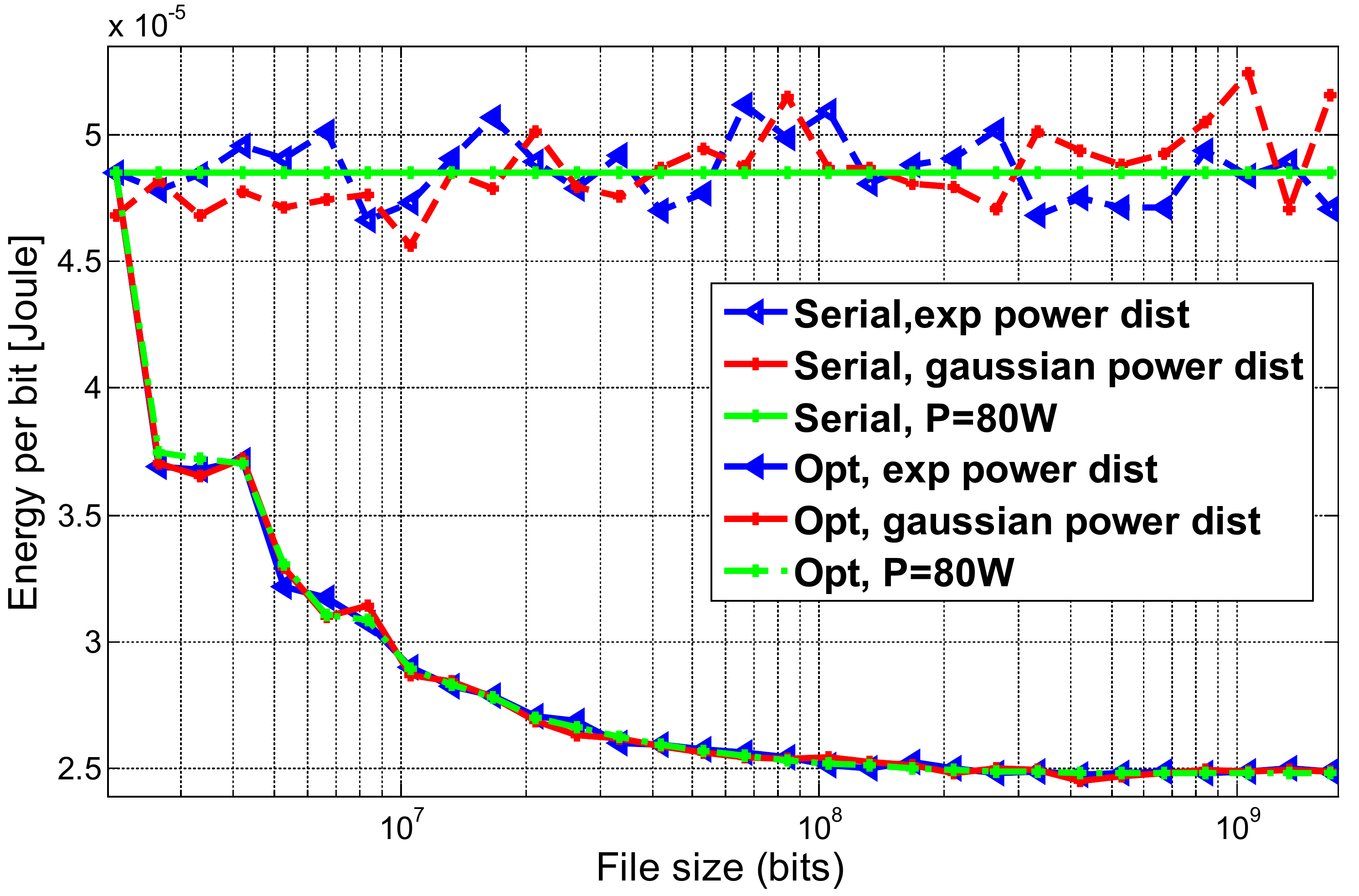}
  \vspace{-0.3in}
    \caption{Impact of heterogeneity in nominal power on the energy per bit consumed by our algorithm, in function of file size, with $200$ hosts.Curves are plotted with $95\%$ confidence interval.}\label{fig:431}
\end{minipage}
\begin{minipage}[b]{0.33\linewidth}\centering
    \includegraphics[width=\textwidth]{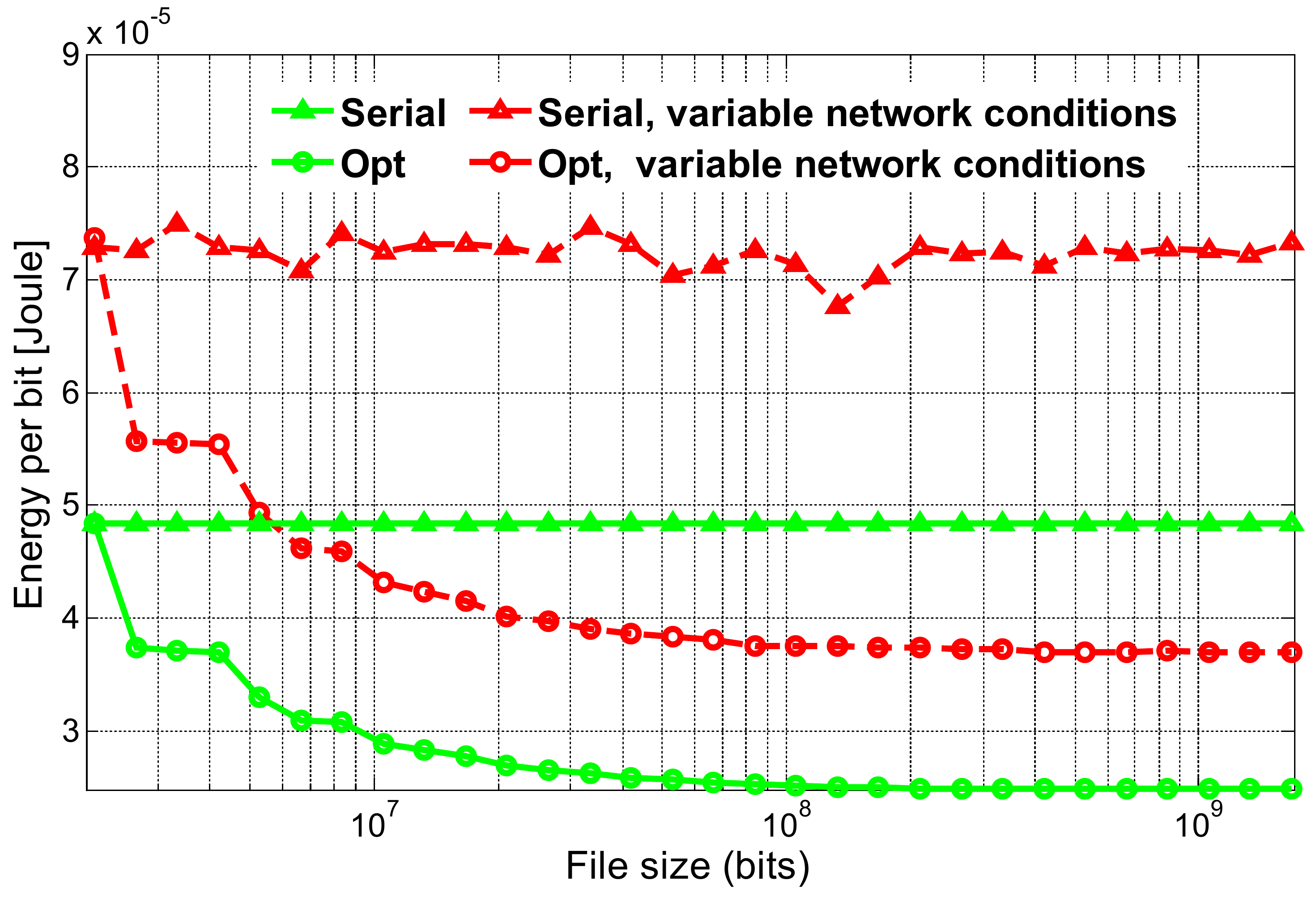}
  \vspace{-0.3in}
    \caption{Impact of variable network conditions on the energy per bit consumed by our algorithm, in function of file size, with $200$ hosts. Curves are plotted with $95\%$ confidence interval.}\label{fig:432}
\end{minipage}

\vspace{-0.15in}
\end{figure*}

\subsubsection{ON/OFF Energy Costs}

As seen in previous sections, our optimal algorithms develop in rounds. Typically, not every host is on in every round (i.e., some go on and off more than once during the file distribution process). In a realistic scenario, a host takes some time to both go off (or into a very low power mode), and to get back to active mode.  Usually, this on/off time is in the order of a few seconds \cite{softpedia2009}. The additional amount of energy consumed while switching between these power states (that we call here ``on/off costs") has potentially an important impact on the energy performance of a scheme, penalizing specifically those schemes in which host activity is more ``discontinuous" over time.

%When on/off costs are present, in order to minimize the impact of on/off costs on its performance, we assume that hosts which should be turned off decide whether to actually turn off or to stay on based on the difference between the time taken to go off and on again, and the estimated time until the next slot in which the host should be on according to the scheme: we assume that when the latter is smaller than the former, the host remains on and waits for the next slot \added[G]{talk about the issue of estimating this time in heterogeneous settings..we shoudl say the algorithm can be made more complex....but this is left for future work}. 

In order to mitigate the negative impact of on/off costs, in our simulations we implement the following mechanism. When a host $A$ has finished its activity (i.e. uploading or/and downloading a block) in an slot $t_1$, and has no activity until slot $t_2$, it computes the energy cost of staying on ($cost_{\mathrm{on}}$) until the slot $t_2$ and the cost of going off during the rest of slot $t_1$ and switching on at the beginning of slot $t_2$ ($cost_{\mathrm{off/on}}$). Hence, if $cost_{\mathrm{on}} \le cost_{\mathrm{off/on}}$, $A$ decides to stay on. Otherwise, it goes off for its non-active period between slots
$t_1$ and $t_2$.

Fig.~\ref{fig:423} presents the energy consumed by our scheme in comparison to the serial scheme considering a switch on/off time equal to $2$ and $4$s. As expected, the on/off costs increase the energy per bit consumed by all schemes. This increment is more pronounced for small file sizes, where we see that on/off costs make the performance of our scheme closer (but still better) to the serial scheme. Conversely, for medium/large file sizes, the contribution of on/off costs to the total energy consumed by a scheme becomes marginal, and the performance of both the optimal scheme and the serial approaches the one in the case without on/off costs. Note the widening of the gap between the serial scheme and our scheme for file sizes around $50$MB is due to the different behavior that our scheme has for the case $n<\beta$ and for the other case.

%To evaluate the effect that on/off costs on our scheme, we have run some experiments in which the time to switch on or off was assumed to be equal to $5$s and $10$s. 
%In  we see how including on/off costs 

\subsubsection{Load Dependency}
\label{sec:loaddep}
In this set of experiments, we have analyzed the impact of the four-states energy model described in \sref{sec:emodel}, which implies some degree of energy proportionality of the host devices. The research community is putting a lot of effort in energy proportionality. Hence, in the future it is expected that network devices will consume energy proportionally to the supported load. Fig.~\ref{fig:424} shows that with the four-states energy model the percentual decrease in the energy per bit consumed by our \emph{Opt} scheme and by the serial one is the same. This suggests that even with load proportional hardware our scheme enables significant energy savings with respect to the serial one. %This promising result suggests that future energy proportional hardware will help to largely reduce the power consumption in file distribution processes.\added[G]{(but of how much? change the plot and show just gains? how parallel ans serial are affected?)}

%From \fref{fig:424} \added[G]{(add plot on serial for energy prop)} we see that the small amount of energy proportionality in our four-states energy model improves the performance of our scheme with respect to the serial one. 

\subsection{Heterogeneous Scenario}
\label{sec:num-het}
In this subsection we consider two separated heterogeneous scenarios. On the one hand, we study the case in which different hosts present different power consumption profiles. On the other hand, we address the scenario in which each host observes different network conditions (i.e., different access speed and congestion level).

\subsubsection{Heterogeneous Power Consumption}
\label{sec-heter}

In Section \ref{sec:homo} we have proved analytically that our \emph{Opt} algorithm minimizes the overall power consumption of the file distribution process, even in a heterogeneous scenario in which each host presents a different energy consumption (as long as all the nodes have the same upload and download rate). To validate this statement, in this subsection we have run experiments in which the nominal power consumed by the hosts varies according to either a Gaussian or an exponential distribution as defined in Section~\ref{sec:emodel}. Then, the energy consumption has been compared with a homogeneous scenario. The results, presented in Fig.~\ref{fig:431}, validate our analysis, since the three curves for the \emph{Opt} scheme overlap perfectly.	We also observe that heterogeneous power consumption has some minor impact in the case of the serial scheme. Finally, it is worth to note that confidence intervals 
have been calculated for each curve (but not shown for clarity), being in any case lower than $5\%$.

%In \fref{fig:431} \added[G]{(revise labels in the legend)}we have the energy per bit consumed by the considered schemes when the nominal power of the hosts varies according to a gaussian or exponential distribution, compared with the performance of the scheme in the homogeneous case, where the nominal power of all hosts is equal to the average nominal power in the heterogeneous case. Curves are plotted with $95\%$ confidence interval (note that, for readability, confidence intervals are plotted only for the gaussian case).
%From this figure we can see that considering the more realistic case of heterogeneity in power consumption among hosts does not change the average performance neither of our scheme, nor of the serial scheme, as none of these schemes takes into account the nominal power of the hosts in determining the schedule for the file distribution process. Of course, in this case better algorithms (but less fair) could be deviced, which for instance maximize the use of the least energy-hungry devices in the scheme, but this is out of the scope of the present paper.    

\subsubsection{Heterogeneous Network Conditions}

In the results presented we have considered $(i)$ similar upload/download access speed for all host and $(ii)$ no network congestion. In this subsection we relax these assumptions, and consider a heterogeneous scenario where hosts have different access speeds and observe different network
state (e.g., congestion). This scenario accurately models a content distribution process in the Internet.

In particular, in the simulations we model the different nominal access speed of hosts using an exponential distribution, based on realistic speed values provided in \cite{akamai_ncdist}. Additionally, in order to model the variation in link speed over time due to network conditions (i.e., congestion) we multiply the nominal access speed by a positive factor taken from a Gaussian distribution with average $1$ and standard deviation $0.07$. Fig.~\ref{fig:432} presents the results for these heterogeneous network conditions, for both our \emph{Opt} scheme and the serial scheme, and compares them with the homogeneous case. The results show that both schemes suffer from an increment in the power consumption, with respect to the homogeneous case. However, the relative difference between the \emph{Opt} and serial schemes increases. This suggests that even in heterogeneous network conditions the proposed algorithm outperforms any centralized scheme.

Moreover, we observe that the energy per bit consumed is constant for both \emph{Opt} and serial schemes when considering heterogeneous network conditions. This occurs because none of the considered schemes takes into account host upload/downlad capacity in determining the schedule for file distribution.

Finally, note that confidence intervals have been obtained for the different curves and all of them present less than $5\%$ difference to the average value in the figure.

\section{Related work}
\label{sec:relatedwork} 

\noindent \textbf{Energy-Efficiency in Networks}:
In order to reduce the overall energy consumption of the Internet, many dimensions for energy savings have been explored. The main efforts include turning off the devices that are unnecessarily on \cite{gupta2003greening,agarwal2009somniloquy}, aggregating traffic streams to send data in bulk \cite{gupta2003greening,nedevschi2008reducing,andrews2010routing1}, network planning \cite{chabarek2008power}, energy efficient routing \cite{restrepo2009energy, andrews2010routing2} and virtualization and migration of routers \cite{wang2008virtual}. Furthermore, some works have addressed specific aspects of energy-efficiency in datacenters \cite{abtsDatacenter,chunDataCenters,ElasticTreesDataCenters}.

\noindent \textbf{Optimization problems in file-distribution processes:}
An important amount of effort has been dedicated to study the completion download time in a file distribution process \cite{kumar2006peer,lingjun2009improving,langner2011optimal}. The minimization of the average finish time in P2P networks is considered in  \cite{sanghavi2007gossiping,matthewtang2009,tsang2010}. Of interest to this paper, \cite{mundinger2008optimal} presents a theoretical study to derive the minimum time associated to a P2P file distribution process. However, an scheme guaranteeing a file distribution with minimum time does not generally leads to minimize the energy consumption. Moreover, schemes with similar distribution time may have different energy costs.
%\added[A][]{It is interesting to note that one of the algorithms presented in this paper is essentially the algorithm used in \cite{sanghavi2007gossiping}
%for a different problem.}

\noindent \textbf{Energy-Efficiency in file distribution:}
To the best of the authors knowledge energy consumption in file distribution processes has received little attention so far. On the one hand, practical studies \cite{ ValanciusNanoDataCenters,lee2011toward,feldmann2010energy, blackburn2009simulation, giuseppe2010} have discussed and compared the energy consumed by different content distribution architectures or protocols. However none of them relies on an analytical basis nor aims to design optimal algorithms, as is the case of our paper. On the other hand,  Mehyar et al. \cite{mehyar2007optimal} and Sucevic et al.\cite{sucevicpowering} (similarly as we do) address the energy-efficiency in file-distribution from an analytical point of view. However, their studies are restricted to P2P schemes whereas the current paper cover both centralized and distributed approaches in order to identify the most efficient scheme. In addition, their analysis is limited to networks of at most 3 nodes. For bigger network sizes, they provide heuristics and use simulations to evaluate energy efficiency. Instead, our analysis is valid for an arbitrary number of nodes. Finally, it is worth to mention that, to the best of our knowledge, we are the first on providing a proof of the NP-hardness of the energy-efficiency optimization problem for file-distribution processes.

\section{Conclusions}
\label{sec:fwop}

This paper presents one of the first dives into a novel and relevant field that has received little attention so far: energy-efficiency in file distribution processes. We present a theoretical framework that constitutes the analytical basis for the design of energy-efficient file distribution protocols. Specifically, this framework reveals two important observations: $(i)$ the general problem of minimizing the energy consumption in a file distribution process is NP-hard and $(ii)$ in all the studied scenarios there exists always a collaborative (i.e. p2p-like) distributed algorithm that reduces the energy consumption of any centralized counterpart. This suggests that in those file distribution processes in which reducing the energy consumption is of significant importance (e.g. software update over night in a corporative network) a distributed algorithm should be implemented. %Indeed, in this paper we present collaborative algorithms that minimize the energy consumption in simple yet realistic scenarios.

\bibliographystyle{IEEEtran}
\bibliography{p2pfiledist}

\begin{thebibliography}{10}

\bibitem{nada}
{\em The NanoDatacenters (NaDa) Project}.

\bibitem{agarwal2009somniloquy}
Y.~Agarwal, S.~Hodges, R.~Chandra, J.~Scott, P.~Bahl, and R.~Gupta.
\newblock Somniloquy: augmenting network interfaces to reduce pc energy usage.
\newblock In {\em Proceedings of the 6th USENIX symposium on Networked systems
  design and implementation}, pages 365--380. USENIX Association, 2009.

\bibitem{andrews2010routing1}
M.~Andrews, A.~Fernandez Anta, L.~Zhang, and W.~Zhao.
\newblock Routing and scheduling for energy and delay minimization in the
  powerdown model.
\newblock In {\em INFOCOM, 2010 Proceedings IEEE}, pages 1--5. IEEE, 2010.

\bibitem{andrews2010routing2}
M.~Andrews, A.~Fernandez Anta, L.~Zhang, and W.~Zhao.
\newblock Routing for energy minimization in the speed scaling model.
\newblock In {\em Transactions of Networking, Accepted for publication, DOI:
  10.1109/TNET.2011.2159864}. IEEE/ACM, 2011.

\bibitem{blackburn2009simulation}
J.~Blackburn and K.~Christensen.
\newblock A simulation study of a new green bittorrent.
\newblock In {\em Communications Workshops, 2009. ICC Workshops 2009. IEEE
  International Conference on}, pages 1--6. IEEE, 2009.

\bibitem{cagalj2002minimum}
M.~{\v{C}}agalj, J.P. Hubaux, and C.~Enz.
\newblock Minimum-energy broadcast in all-wireless networks: Np-completeness
  and distribution issues.
\newblock In {\em Proceedings of the 8th annual international conference on
  Mobile computing and networking}, pages 172--182. ACM, 2002.

\bibitem{chabarek2008power}
J.~Chabarek, J.~Sommers, P.~Barford, C.~Estan, D.~Tsiang, and S.~Wright.
\newblock Power awareness in network design and routing.
\newblock In {\em INFOCOM 2008. The 27th Conference on Computer Communications.
  IEEE}, pages 457--465. IEEE, 2008.

\bibitem{dale2009apt}
C.~Dale and J.~Liu.
\newblock apt-p2p: A peer-to-peer distribution system for software package
  releases and updates.
\newblock In {\em INFOCOM 2009, IEEE}, pages 864--872. IEEE, 2009.

\bibitem{feldmann2010energy}
A.~Feldmann, A.~Gladisch, M.~Kind, C.~Lange, G.~Smaragdakis, and F.J. Westphal.
\newblock Energy trade-offs among content delivery architectures.
\newblock In {\em Telecommunications Internet and Media Techno Economics
  (CTTE), 2010 9th Conference on}, pages 1--6. IEEE, 2010.

\bibitem{fragouli2006network}
C.~Fragouli, J.~Widmer, and J.Y.L. Boudec.
\newblock A network coding approach to energy efficient broadcasting: from
  theory to practice.
\newblock In {\em IEEE Infocom}, 2006.

\bibitem{gkantsidis2006planet}
C.~Gkantsidis, T.~Karagiannis, and M.~VojnoviC.
\newblock Planet scale software updates.
\newblock In {\em Proceedings of the 2006 conference on Applications,
  technologies, architectures, and protocols for computer communications},
  pages 423--434. ACM, 2006.

\bibitem{gupta2003greening}
M.~Gupta and S.~Singh.
\newblock Greening of the internet.
\newblock In {\em Proceedings of the 2003 conference on Applications,
  technologies, architectures, and protocols for computer communications},
  pages 19--26. ACM, 2003.

\bibitem{kumar2006peer}
R.~Kumar and K.W. Ross.
\newblock Peer-assisted file distribution: The minimum distribution time.
\newblock In {\em IEEE Workshop on Hot Topics in Web Systems and Technologies
  (HOTWEB’06)}, pages 1--11, 2006.

\bibitem{labovitz2010internet}
C.~Labovitz, S.~Iekel-Johnson, D.~McPherson, J.~Oberheide, and F.~Jahanian.
\newblock Internet inter-domain traffic.
\newblock {\em ACM SIGCOMM Computer Communication Review}, 40(4):75--86, 2010.

\bibitem{langner2011optimal}
T.~Langner, C.~Schindelhauer, and A.~Souza.
\newblock Optimal file-distribution in heterogeneous and asymmetric storage
  networks.
\newblock {\em SOFSEM 2011: Theory and Practice of Computer Science}, pages
  368--381, 2011.

\bibitem{lee2011toward}
U.~Lee, I.~Rimac, D.~Kilper, and V.~Hilt.
\newblock Toward energy-efficient content dissemination.
\newblock {\em Network, IEEE}, 25(2):14--19, 2011.

\bibitem{lingjun2009improving}
M.~Lingjun, P.S. Tsang, and K.S. Lui.
\newblock Improving file distribution performance by grouping in peer-to-peer
  networks.
\newblock {\em Network and Service Management, IEEE Transactions on},
  6(3):149--162, 2009.

\bibitem{mehyar2007optimal}
M.~Mehyar, W.H. Gu, S.H. Low, M.~Effros, and T.~Ho.
\newblock Optimal strategies for efficient peer-to-peer file sharing.
\newblock In {\em Acoustics, Speech and Signal Processing, 2007. ICASSP 2007.
  IEEE International Conference on}, volume~4, pages IV--1337. IEEE, 2007.

\bibitem{mundinger2008optimal}
J.~Mundinger, R.~Weber, and G.~Weiss.
\newblock Optimal scheduling of peer-to-peer file dissemination.
\newblock {\em Journal of Scheduling}, 11(2):105--120, 2008.

\bibitem{nedevschi2008reducing}
S.~Nedevschi, L.~Popa, G.~Iannaccone, S.~Ratnasamy, and D.~Wetherall.
\newblock Reducing network energy consumption via sleeping and rate-adaptation.
\newblock In {\em Proceedings of the 5th USENIX Symposium on Networked Systems
  Design and Implementation}, pages 323--336. USENIX Association, 2008.

\bibitem{restrepo2009energy}
J.C.C. Restrepo, C.G. Gruber, and C.M. Machuca.
\newblock Energy profile aware routing.
\newblock In {\em Communications Workshops, 2009. ICC Workshops 2009. IEEE
  International Conference on}, pages 1--5. IEEE, 2009.

\bibitem{sanghavi2007gossiping}
S.~Sanghavi, B.~Hajek, and L.~Massoulie.
\newblock Gossiping with multiple messages.
\newblock {\em Information Theory, IEEE Transactions on}, 53(12):4640--4654,
  2007.

\bibitem{sucevicpowering}
A.~Sucevic, L.L.H. Andrew, and T.T.T. Nguyen.
\newblock Powering down for energy efficient peer-to-peer file distribution.
\newblock GreenMetrics, 2011.

\bibitem{wang2008virtual}
Y.~Wang, E.~Keller, B.~Biskeborn, J.~van~der Merwe, and J.~Rexford.
\newblock Virtual routers on the move: live router migration as a
  network-management primitive.
\newblock In {\em ACM SIGCOMM Computer Communication Review}, volume~38, pages
  231--242. ACM, 2008.

\bibitem{widmer2005low}
J.~Widmer, C.~Fragouli, and J.Y. Le~Boudec.
\newblock Low-complexity energy-efficient broadcasting in wireless ad-hoc
  networks using network coding.
\newblock In {\em Proc. Workshop on Network Coding, Theory, and Applications},
  2005.

\end{thebibliography}


\begin{thebibliography}{10}
\providecommand{\url}[1]{#1}
\csname url@samestyle\endcsname
\providecommand{\newblock}{\relax}
\providecommand{\bibinfo}[2]{#2}
\providecommand{\BIBentrySTDinterwordspacing}{\spaceskip=0pt\relax}
\providecommand{\BIBentryALTinterwordstretchfactor}{4}
\providecommand{\BIBentryALTinterwordspacing}{\spaceskip=\fontdimen2\font plus
\BIBentryALTinterwordstretchfactor\fontdimen3\font minus
  \fontdimen4\font\relax}
\providecommand{\BIBforeignlanguage}[2]{{%
\expandafter\ifx\csname l@#1\endcsname\relax
\typeout{** WARNING: IEEEtran.bst: No hyphenation pattern has been}%
\typeout{** loaded for the language `#1'. Using the pattern for}%
\typeout{** the default language instead.}%
\else
\language=\csname l@#1\endcsname
\fi
#2}}
\providecommand{\BIBdecl}{\relax}
\BIBdecl

\bibitem{parliament}
\BIBentryALTinterwordspacing
\emph{ICT and CO2 emissions}. [Online]. Available:
  \url{http://www.parliament.uk/documents/post/postpn319.pdf}
\BIBentrySTDinterwordspacing

\bibitem{pickavetEnergyConsumption}
M.~Pickavet, W.~Vereecken, S.~Demeyer, P.~Audenaert, B.~Vermeulen, C.~Develder,
  D.~Colle, B.~Dhoedt, and P.~Demeester, ``Worldwide energy needs for ict: The
  rise of power-aware networking,'' in \emph{ANTS}, 2008.

\bibitem{marcoMeliaEnergyCostISP}
L.~Chiaraviglio, M.~Mellia, and F.~Neri, ``Minimizing isp network energy cost:
  Formulation and solutions,'' \emph{IEEE/ACM Transactions on Networking,},
  2011.

\bibitem{NikosDSLAMSigcom}
E.~Goma, M.~Canini, A.~Lopez~Toledo, N.~Laoutaris, D.~Kosti\'{c}, P.~Rodriguez,
  R.~Stanojevi\'{c}, and P.~Yag\"{u}e~Valentin, ``Insomnia in the access: or
  how to curb access network related energy consumption,'' in \emph{ACM
  SIGCOMM}, 2011.

\bibitem{ElasticTreesDataCenters}
B.~Heller, S.~Seetharaman, P.~Mahadevan, Y.~Yiakoumis, P.~Sharma, S.~Banerjee,
  and N.~McKeown, ``Elastictree: saving energy in data center networks,'' in
  \emph{NSDI}, 2010.

\bibitem{ValanciusNanoDataCenters}
V.~Valancius, N.~Laoutaris, L.~Massouli\'{e}, C.~Diot, and P.~Rodriguez,
  ``Greening the internet with nano data centers,'' in \emph{ACM CoNEXT}, 2009.

\bibitem{bolla2011}
R.~Bolla, R.~Bruschi, F.~Davoli, and F.~Cucchietti, ``Energy efficiency in the
  future internet: A survey of existing approaches and trends in energy-aware
  fixed network infrastructures,'' \emph{Communications Surveys \& Tutorials,
  IEEE}, vol.~13, no.~2, pp. 223--244, 2011.

\bibitem{restrepo2009energy}
J.~Restrepo, C.~Gruber, and C.~Machuca, ``Energy profile aware routing,'' in
  \emph{Communications Workshops, IEEE ICC 2009.}, 2009, pp. 1--5.

\bibitem{andrews2010routing2}
M.~Andrews, A.~Fern\'andez~Anta, L.~Zhang, and W.~Zhao, ``Routing for energy
  minimization in the speed scaling model,'' in \emph{Transactions of
  Networking, Accepted for publication in 2011, DOI:
  10.1109/TNET.2011.2159864}, 2010.

\bibitem{gupta2003greening}
M.~Gupta and S.~Singh, ``Greening of the internet,'' in \emph{SIGCOMM}, 2003.

\bibitem{agarwal2009somniloquy}
Y.~Agarwal, S.~Hodges, R.~Chandra, J.~Scott, P.~Bahl, and R.~Gupta,
  ``Somniloquy: augmenting network interfaces to reduce pc energy usage,'' in
  \emph{NSDI}, 2009, pp. 365--380.

\bibitem{gkantsidis2006planet}
C.~Gkantsidis, T.~Karagiannis, and M.~VojnoviC, ``Planet scale software
  updates,'' in \emph{ACM SIGCOMM}, 2006.

\bibitem{LabovitzInternetTraffic}
C.~Labovitz, S.~Iekel-Johnson, D.~McPherson, J.~Oberheide, and F.~Jahanian,
  ``Internet inter-domain traffic,'' in \emph{ACM SIGCOMM}, 2010.

\bibitem{sandvineReportFall2011}
``Sandvine fall 2011 global internet phenomena report,''
  http://www.sandvine.com/news/global\_broadband\_trends.asp.

\bibitem{NordmanC09}
B.~Nordman and K.~J. Christensen, ``Greener pcs for the enterprise,'' \emph{IT
  Professional}, vol.~11, no.~4, pp. 28--37, 2009.

\bibitem{softpedia2009}
``"in windows 7 use sleep to resume the os in 2 seconds",''
  http://news.softpedia.com/news/In-Windows-7-Use-Sleep-to-Resume-the-OS-in-2-%
Seconds-101290.shtml.

\bibitem{akamai_ncdist}
``The real connection speeds for internet users across the world (charts),''
  http://royal.pingdom.com/2010/11/12/real-connection-speeds-for-internet-user%
s-across-the-world/.

\bibitem{nedevschi2008reducing}
S.~Nedevschi, L.~Popa, G.~Iannaccone, S.~Ratnasamy, and D.~Wetherall,
  ``Reducing network energy consumption via sleeping and rate-adaptation,'' in
  \emph{NSDI}, 2008.

\bibitem{andrews2010routing1}
M.~Andrews, A.~Fern\'andez~Anta, L.~Zhang, and W.~Zhao, ``Routing and
  scheduling for energy and delay minimization in the powerdown model,'' in
  \emph{IEEE INFOCOM}, 2010.

\bibitem{chabarek2008power}
J.~Chabarek, J.~Sommers, P.~Barford, C.~Estan, D.~Tsiang, and S.~Wright,
  ``Power awareness in network design and routing,'' in \emph{INFOCOM}, 2008,
  pp. 457--465.

\bibitem{wang2008virtual}
Y.~Wang, E.~Keller, B.~Biskeborn, J.~van~der Merwe, and J.~Rexford, ``Virtual
  routers on the move: live router migration as a network-management
  primitive,'' in \emph{ACM SIGCOMM}, 2008.

\bibitem{abtsDatacenter}
D.~Abts, M.~R. Marty, P.~M. Wells, P.~Klausler, and H.~Liu, ``Energy
  proportional datacenter networks,'' \emph{SIGARCH Comput. Archit. News},
  vol.~38, pp. 338--347, June 2010.

\bibitem{chunDataCenters}
B.-G. Chun, G.~Iannaccone, G.~Iannaccone, R.~Katz, G.~Lee, and L.~Niccolini,
  ``An energy case for hybrid datacenters,'' \emph{SIGOPS Oper. Syst. Rev.},
  vol.~44, pp. 76--80, March 2010.

\bibitem{kumar2006peer}
R.~Kumar and K.~Ross, ``Peer-assisted file distribution: The minimum
  distribution time,'' in \emph{IEEE Workshop on Hot Topics in Web Systems and
  Technologies (HOTWEB 06)}, 2006, pp. 1--11.

\bibitem{lingjun2009improving}
M.~Lingjun, P.~Tsang, and K.~Lui, ``Improving file distribution performance by
  grouping in peer-to-peer networks,'' \emph{IEEE Transactions on Network and
  Service Management}, vol.~6, no.~3, pp. 149--162, 2009.

\bibitem{langner2011optimal}
T.~Langner, C.~Schindelhauer, and A.~Souza, ``Optimal file-distribution in
  heterogeneous and asymmetric storage networks,'' \emph{SOFSEM 2011: Theory
  and Practice of Computer Science}, pp. 368--381, 2011.

\bibitem{sanghavi2007gossiping}
S.~Sanghavi, B.~Hajek, and L.~Massoulie, ``Gossiping with multiple messages,''
  \emph{IEEE Transactions on Information Theory}, vol.~53, no.~12, pp.
  4640--4654, 2007.

\bibitem{matthewtang2009}
L.~L.~A. G.~Matthew~Ezovski, Ao~Tang, ``Minimizing average finish time in p2p
  networks,'' in \emph{IEEE Infocom}, 2009.

\bibitem{tsang2010}
K.-S.~L. Pui-Sze~Tsang, Xiang~Meng, ``A novel grouping strategy for reducing
  average distribution time in p2p file sharing,'' in \emph{IEEE ICC}, 2010.

\bibitem{mundinger2008optimal}
J.~Mundinger, R.~Weber, and G.~Weiss, ``Optimal scheduling of peer-to-peer file
  dissemination,'' \emph{Journal of Scheduling}, vol.~11, no.~2, pp. 105--120,
  2008.

\bibitem{lee2011toward}
U.~Lee, I.~Rimac, D.~Kilper, and V.~Hilt, ``Toward energy-efficient content
  dissemination,'' \emph{Network, IEEE}, vol.~25, no.~2, pp. 14--19, 2011.

\bibitem{feldmann2010energy}
A.~Feldmann, A.~Gladisch, M.~Kind, C.~Lange, G.~Smaragdakis, and F.~Westphal,
  ``Energy trade-offs among content delivery architectures,'' in \emph{IEEE
  Telecommunications Internet and Media Techno Economics (CTTE),}, 2010, pp.
  1--6.

\bibitem{blackburn2009simulation}
J.~Blackburn and K.~Christensen, ``A simulation study of a new green
  bittorrent,'' in \emph{Communications Workshops, ICC}, 2009, pp. 1--6.

\bibitem{giuseppe2010}
A.~P. Giuseppe~Anastasi, Ilaria~Giannetti, ``A bittorrent proxy for green
  internet file sharing: Design and experimental evaluation,'' \emph{Computer
  Communications}, vol.~33, no.~7, pp. 794--802, 2010.

\bibitem{mehyar2007optimal}
M.~Mehyar, W.~Gu, S.~Low, M.~Effros, and T.~Ho, ``Optimal strategies for
  efficient peer-to-peer file sharing,'' in \emph{Acoustics, Speech and Signal
  Processing, ICASSP}, vol.~4, 2007.

\bibitem{sucevicpowering}
A.~Sucevic, L.~Andrew, and T.~Nguyen, ``Powering down for energy efficient
  peer-to-peer file distribution,'' 2011.

\end{thebibliography}
\appendix
\label{appendix}

\newtheorem{observation}{Observation}

\subsection{NP-hardness}
\label{sec:np-hard}

We show in the section that a general version of the problem considered in this paper is NP-hard. 
%In this version of the problem, hosts have different upload and download capacities, different power consumptions, and different block sizes. 
The following theorem summarizes the result.

\begin{thm}
Assume that time is slotted, that hosts must upload at their full capacity, and that no host can upload to more than one host in the same slot.
The problem of minimizing the energy of file distribution is NP-hard if hosts can have different upload capacities and power consumptions, even if
$\alpha_i=\delta_i=0, \forall i$.
\end{thm}
\begin{proof}
We use reduction from the partition problem. 
The input of this problem is a set of integers (we assume all of them to be positive) $A=\{x_0, x_2, ..., x_{k-1}\}$, $k >1$. Let $M=\sum_{x_i \in A} x_i$ to be even. The problem is to decide whether there is a subset $A' \subset A$ such that $\sum_{x_i \in A'} x_i=M/2$.

We reduce an instance of the partition problem to an instance of our problem as follows. The file to distribute has $M$ blocks of size $1$.
There are $n=k+3$ hosts: server $S$, hosts $T$ and $R$, and hosts $H_i$, for $i \in [0,k-1]$.
All hosts have fixed setup energy $\delta_i=0$ and no cost for switching on and off, i.e., $\alpha_i=0$.
Server $S$ has upload capacity $M$ and power $P$.
Host $T$ has download and upload capacity $M$, and power $P$.
Hosts $H_i$, $i \in [0,k-1]$, have download capacity $M$, upload capacity $u_i=x_i$, and power consumption $P$.
Host $R$ has download capacity $M/2$ and power consumption $P' >2P(2k+1)$.
The slot length is one unit of time.

Observe that there is always a feasible solution that respects the assumptions of the model. It works as follows.
First, $S$ serves the whole file to $T$ in one slot. 
Then, $T$ serves the whole file to hosts $H_i$, $i \in [0,k-1]$, in consecutive slots.
Finally, each host $H_i$, $i \in [0,k-1]$, serves $x_i$ different blocks to $R$ in consecutive slots.

We claim that the subset $A'$ that satisfies $\sum_{x_i \in A'} x_i=M/2$ exists if and only if the file distribution problem can be solved with energy smaller than
$3P'$. Hence, the energy minimization problem is NP-hard.

If subset $A'$ exists, the following schedule is feasible.
First, $S$ serves $T$ the whole file in one slot. 
Then, $T$ serves each host $H_i$, $i \in [0,k-1]$, the whole file in consecutive slots.
Let $U=\cup_{x_i \in A'} \{H_i\}$, then the hosts in $U$ upload the file to $R$ in two slots, half the file in each slot.
The total energy consumed is 
$$
E=2P+2Pk+2(|A'|P + P') \leq 2P(2k+1)+2P'<3P'.
$$

Assume now that there is a schedule with energy less than $3P'$. Then, $R$ has been up two slots. Since
they cannot upload at full capacity to $R$, and they cannot serve more than one host, neither $S$ nor $T$ can serve $R$. 
Then, looking at the first slot in which $R$ is up, $R$ must have been served by a subset of hosts $H_i$ whose aggregate 
upload capacity is exactly $M/2$. This proves the existence of $A'$.
\end{proof}

%--------------------------------------------------------------------------------------------------------------------------------------------------------------------------------------------

\subsection{Proof of Theorem~\ref{thmeq}}
\label{thmeq-proof}

We transform the cost of a block as defined in Equation~\ref{costblock} to the following one. For each host $i \in \mathcal{I}_{\tau}^z$, define $\phi_i$ and $\psi_i$ as 

\[\phi_{i} = \left\{ 
\begin{array}{l l}
  \Delta_i & \quad \mbox{if $\mathcal{S}_{i,\tau}^z\neq \emptyset$ }\\
% 	& \quad \mbox{$w=$min$\{u|(u,v) \in E, d(u)=0\}$}\\
  0 & \quad \mbox{Otherwise}\\ \end{array} \right. \]

\[\psi_{i} = \left\{ 
\begin{array}{l l}
  \Delta_i & \quad \mbox{if $\mathcal{S}_{i,\tau}^z=\emptyset$ }\\
  0 & \quad \mbox{Otherwise}\\ \end{array} \right. \]

Note that $\sum_{b_j\in \mathcal{S}_{i,\tau}^z} \mathcal{D}_{j,i}^z =1$ iff $|\mathcal{S}_{i,\tau}^z|\geq 1$ (i.e., when $\phi_i = \Delta_i$). 
It is easy to see that $\mathcal{U}_{j,i}^z =1$ iff $\psi_{serv(j,i)} = \Delta_{serv(j,i)}$, i.e., $S_{{serv(j,i)},\tau}^z=\emptyset$. 
Therefore, for a host $i \in \mathcal{I}_{\tau}^z$, either $\phi_i=\Delta_i$ or $\psi_i=\Delta_i$, never both 0 or both $\Delta_i$. Hence,

\[ \sum_{i\in \mathcal{I}_{\tau}^z} (\phi_i + \psi_i) = \sum_{i\in\mathcal{I}_{\tau}^z}\Delta_i \]

%--------------------------------------------------------------------------------------------------------------------------------------------------------------------

\subsection{Proof of Theorem \ref{minhetero}  (Lower Bound for $k=1$)}
\label{s-lb-du}

The claim to shown is that if $k=1$ any scheme $z$ consumes energy
\begin{eqnarray}
\label{eq:minhetero}
E(z) \geq \beta \left( \Delta_S + \sum_{i=0}^{n-1} \Delta_i \right) +  \max\{0,n-\beta\} \min\{\Delta_S,\Delta_0\}
\label{eqn:opt1hetero}
\end{eqnarray}

Before proving the claim, we need some supporting claims.

\begin{lem}
For every block $b_j$ and every client $H_i$ it holds that $\mathcal{D}_{j,i}^z=1$.
\label{lem:noblock0}
\end{lem}
\begin{proof}
Since $d=u$, each host can receive only one block in a time slot. Hence, if block $b_j$ is transferred to client $H_i$ in slot $\tau$, we have $|\mathcal{S}_{i,\tau}^z|=1$. Then, by definition, 
$\mathcal{D}_{j,i}^z=1$.
\end{proof}

\begin{lem}
For every block $b_j$ served by $S$ to client $H_i$, it holds $\mathcal{U}_{j,i}^z = 1$. 
\label{lem:server2}
\end{lem}

\begin{proof}
%This follows from the definition of the cost of a block. 
Let $S$ be serving $b_j$ to $H_i$ in slot $\tau$. 
%From the previous lemma 
%, hence $|\mathcal{S}_{i,\tau}^z|=1$
%$\forall z\in \mathcal{\hat Z}^{n,\beta}_1 $, $\forall i\in\{0,1,..,n-1\}, \forall \tau \in \{1,2, .., \tau_f^z \}$. Hence, 
%$\mathcal{D}_{j,i}^z=1$.
% for $\forall j\in \{0,1,..,\beta-1\}$. 
Then, $\mathcal{S}_{S,\tau}^z$ is always $\emptyset$, because the server never receives any block from the clients, which means that $\mathcal{U}_{j,i}^z=1$ for any block $b_j$ served by $S$.  
\end{proof}

Since $S$ has to serve each block of the file at least once, we obtain the following corollary.
\begin{cor}
\label{cor1}
For at least $\beta$ block transfers $\mathcal{U}_{j,i}^z =1$.
% because the server has to serve at least $\beta$ slots. 
%Furthermore, $\mathcal{U}_{j,i}^z=1$ $\forall i,j | H_{serv(j,i)}^z = S$
\end{cor}
%
%\begin{lem}
%In case $d=u$, there is no block $b_j$ with $\mathcal{U}_{j,i}^z + \mathcal{D}_{j,i}^z=0$. Moreover, $\mathcal{D}_{j,i}^z=1$ $\forall z\in\mathcal{\hat Z}^{n,\beta}_1,$ $\forall i,j$.
%\label{lem:noblock0}
%\end{lem}
%\begin{proof}
%For every block that is being transfered, we have $|\mathcal{S}_{i,\tau}^z|=1$ for the receiving host $H_i$. Hence, $\mathcal{D}_{j,i}^z=1$ for every $b_j$ for all $H_i$.
%\end{proof}

\begin{lem}
If there exists a host $H$ that is receiving its first block in a time slot $\tau$, then there is at least one block $b_j$ in $\tau$ such that $\mathcal{U}_{j,i}^z =1$.  
\label{lem:starting2}
\end{lem}
\begin{proof}
The number of active hosts in slot $\tau$ is $|\mathcal{I}^z_{\tau}|$.  At most $|\mathcal{I}^z_{\tau}|-1$ blocks can be transferred in $\tau$ because host $H$ cannot upload to anyone. 
Then, since $d=u$, there exists at least one host $H'$ that is on only for uploading. Let $b_j$ be the block served by $H'$. 
%From Lemma \ref{lem:noblock0}, $\mathcal{D}_{j,i}^z=1$.
As it is not downloading any block, $\mathcal{S}_{H',\tau}^z=\emptyset$ and hence $\mathcal{U}_{j,i}^z=1$. 
\end{proof}

\begin{cor}
\label{cor2}
There are $n$ hosts that receive a block for the first time. Thus, for at least $n$ block transfers $\mathcal{U}_{j,i}^z=1$.
%, $ \forall z\in \mathcal{\hat Z}^{n,\beta}_1$. 
\end{cor}

We now prove the claim.
In order to compute the minimum energy consumption, we need to lower bound Equation \ref{costschblock}. 
From Lemma \ref{lem:noblock0}, it follows that 
\begin{equation}
\label{eq:eq1proof}
\sum_{i=0}^{n-1}\sum_{j=0}^{\beta-1} \Delta_i\cdot\mathcal{D}_{j,i}^z =  \beta \cdot\sum_{i=0}^{n-1}\Delta_i
\end{equation} 

from Lemma \ref{lem:server2} and Corollaries \ref{cor1} and \ref{cor2},
%\begin{equation}
%\sum_{i=0}^{n-1}\sum_{j=0}^{\beta-1} \mathcal{U}_{j,i}^z \geq \max\{n,\beta\}
%\end{equation}
%
%Thus,
\begin{equation}
\label{eq:eq2proof}
\sum_{i=0}^{n-1}\sum_{j=0}^{\beta-1} \Delta_{serv(j,i)} \cdot \mathcal{U}_{j,i}^z \geq \beta\cdot\Delta_S + \max\{0,n-\beta\}\cdot\min\{\Delta_S,\Delta_0\}
\end{equation}

Adding Equations \ref{eq:eq1proof} and \ref{eq:eq2proof}, the claim follows.

%-------------------------------------------

\subsection{Proofs of Correctness and Optimality for $k=1$}
\label{s-algo-proofs}

For the correctness and optimality proofs of a scheme $z$ (described by an algorithm),
%of Algorithm \ref{algo1} and \ref{algo2}, Section \ref{sec:homo}, 
we define the state $\sigma_{i,\tau}^z$ of a host $i \in {\cal I}$ at the end of slot $\tau$ as the set of blocks held by that time at the host. Thus, to start with, initially for $S$ we have, 
%\begin{equation} 
%\sigma_{S,0}^z = {\cal B},
%%\bigcup_{j=0}^{\beta-1}\{ b_j \}
%\end{equation}
%and, for each client $i \in \{0,...,n-1\}$, 
%\begin{equation}
%\sigma_{i,0}^z = \emptyset
%\end{equation}
%% where $\vec{0} = (0,0,..,0,..,0)$. 
%
%If $z$ is correct, after the makespan of $z$ ($\tau_f^z$ slots) the state of every client $i\in\{0,...,n-1\}$ must be 
%\begin{equation} 
%\sigma_{i,\tau_f^z}^z = {\cal B}
%%\bigcup_{j=0}^{\beta-1} \{ b_j\}
%\end{equation}
$\sigma_{S,0}^z = {\cal B}$, and, for each client $i \in \{0,...,n-1\}$, 
$\sigma_{i,0}^z = \emptyset$.
If $z$ is correct, after the makespan of $z$ ($\tau_f^z$ slots) the state of every client $i\in\{0,...,n-1\}$ must be 
$\sigma_{i,\tau_f^z}^z = {\cal B}$.
We omit $z$ and $\tau$ when clear from the context.

\subsubsection{Algorithm \ref{algo1}}

% \paragraph{Correctness}

Let us denote the scheme described by Algorithm \ref{algo1} as $z_1$. This scheme has the following properties.

\begin{observation} 
After the \emph{for} loop at Lines \ref{alg2:fl1}-\ref{alg2:fl1end},  the state of client $i$ is $\sigma_i=\{b_i\}, \forall i\in \{0,..,n-1\}$.
%\[\sigma_i=\{b_i\} \]
\end{observation}

\begin{lem}
\label{proof1:lem1}
After the $q^{\mathrm{th}}$ iteration of the loop at Lines \ref{alg1:fl2}-\ref{alg1:fl2end}, for $q\in\{0,...,n-1 \}$, each host $H_i$, $i\in\{ 0,...,n-1\}$ has state 
\begin{equation}
\sigma_i=\bigcup_{p=0}^q \{b_{(i+p)\bmod{n}}\}
\end{equation}
\end{lem}

\begin{proof}
We prove the claim by induction on $q$. 
The base case ($q=0$) holds from the observation: After the \emph{for} loop at lines \ref{alg1:fl1}-\ref{alg1:fl1end},
%\begin{equation}
$\sigma_i = \{b_i \}$.
%\end{equation}

Assuming the hypothesis to be true for $q-1$, in the $q^{\mathrm{th}}$ iteration $H_i$ receives block $b_{(i+j+1)\bmod{n}}$. In this iteration, the value of $j$ is $j=n+q-1$. Hence, $H_i$ receives $b_{(i+q)\bmod{n}}$, and the state after the $q^{\mathrm{th}}$ iteration is
\begin{equation}
\sigma_i= \bigcup_{p=0}^{q-1} \{b_{(i+p)}\} \cup \{b_{(i+q)\bmod{n}} \}=\bigcup_{p=0}^q \{b_{(i+p)\bmod{n}}\}
\end{equation}
\end{proof}

\begin{lem}
\label{proof1:lem2}
In every iteration of the \emph{for} loop at Lines \ref{alg1:fl2}-\ref{alg1:fl2end}, host $H_i, i\in\{0,..,n-1 \}$ serves one of the blocks it has already downloaded. 
\end{lem}

\begin{proof}
In the $q^{\mathrm{th}}$ iteration, $q \geq 1$, $H_i$ serves block $b_{(i+j)\bmod{n}} = b_{(i+q+n-1)\bmod{n}} = b_{(i+q-1)\bmod{n}}$. From the previous lemma, after the $(q-1)^{\mathrm{th}}$ iteration,
the state of $i$ is
\begin{equation}
\sigma_i=\bigcup_{p=0}^{q-1}\{b_{(i+p)\bmod{n}} \}
\end{equation}
which includes $b_{(i+q-1)\bmod{n}}$. Hence the claim follows.
\end{proof}

\begin{thm}
After the termination of Algorithm \ref{algo1} each client $H_i$, $i \in\{0,...,n-1 \}$, has received all the blocks $b_j \in {\cal B}$ with optimal energy $E(z_1)=n ( \Delta_S + \sum_{i=0}^{n-1} \Delta_i )$. 
\end{thm}
\begin{proof}
It follows from Lemma \ref{proof1:lem1} that after the $(n-1)^{\mathrm{th}}$ iteration of the loop at Lines \ref{alg1:fl2}-\ref{alg1:fl2end}, each host has received all the blocks. The scheme is then correct, since each host serves a block it has already downloaded (Lemma \ref{proof1:lem2}). Each host (including the server) is active exactly $n$ slots. Then,
the total energy consumed is $E(z_1)=n ( \Delta_S + \sum_{i=0}^{n-1} \Delta_i )$, which is optimal since it matches the lower bound.
\end{proof}

\subsubsection{Algorithm \ref{algo2}}
% Same as Algorithm \ref{algo1}.

Let us denote the scheme described by Algorithm \ref{algo2} as $z_2$. This scheme has the following properties.

\begin{observation}
After the \emph{for} loop at Lines \ref{alg2:fl1}-\ref{alg2:fl1end}, the state of client $i$ is $\sigma_i=\{b_i\}, \forall i\in \{0,..,n-1\}$.
\end{observation}

\begin{lem}
\label{proof2:lem1}
After the $q^{\mathrm{th}}$ iteration of the loop at Lines \ref{alg2:fl2}-\ref{alg2:fl2end}, for $q\in \{0,1,..,\beta-n\}$, each host 
$H_i, i\in\{0,...,n-1\}$, has state 
\begin{equation}
\sigma_i = \bigcup_{p=0}^{q} \{b_{(i+p)}\}
\end{equation}
\end{lem}

\begin{proof}
We use induction on $q$ to prove the lemma. The base case ($q=0$) follows from the observation. 

Induction step: Assume the hypothesis to be true for the $(q-1)^{\mathrm{th}}$ iteration. Client $H_i, i\in\{0,...,n-2\}$ receives block $b_{(i+q)}$ in the $q^{\mathrm{th}}$ iteration, while client $H_{n-1}$ receives block $b_{(q+n-1)}$ from the server. Thus, $\forall i\in\{0,...,n-1\}$, the state of client $H_i$ after the $q^{\mathrm{th}}$ iteration is
\[\sigma_i=\bigcup_{p=0}^{q-1} \{b_{(i+p)} \}  \cup \{b_{(i+q)} \} = \bigcup_{p=0}^{q} \{b_{(i+p)}\}
\]
\end{proof}

\begin{lem}
\label{proof2:lem2}
After the $q'^{\mathrm{th}}$ iteration of the loop at Lines \ref{alg2:fl3}-\ref{alg2:fl3end}, for $q'\in\{0,1,..,n-1\}$, each host $H_i, i\in\{0,1,..,n-1\}$, has state 
\begin{equation}
\sigma_i = \bigcup_{p=0}^{q'+\beta-n} \{b_{(i+p)\bmod{\beta}}\}
\end{equation}
\end{lem}
\begin{proof}
We use induction on $q'$ to prove the claim. The base case ($q'=0$) follows from Lemma \ref{proof2:lem1} with $q=\beta-n$. Let the claim (induction hypothesis) be true for the $(q'-1)^{\mathrm{th}}$ iteration. In the $q'^{\mathrm{th}}$ iteration, the value of $j$ is $j=q'+\beta-1$. Hence, $H_i$ receives block $b_{(i+q'+\beta-n)}$. Thus, the state of client $H_i$ after the $q'^{\mathrm{th}}$ iteration is
\begin{eqnarray}
\sigma_i & = &\bigcup_{p=0}^{q'-1+\beta-n} \{b_{(i+p)\bmod{\beta}}\} \cup \{ b_{(i+q'+\beta-n)\bmod{\beta}}\} \nonumber \\
& = & \bigcup_{p=0}^{q'+\beta-n} \{b_{(i+p)\bmod{\beta}}\}
\end{eqnarray}
\end{proof}

\begin{lem}
\label{proof2:lem3}
During the execution of Algorithm \ref{algo2} each host $H_i, i\in\{0,...,n-1\}$ serves a block that it has already downloaded. 
\end{lem}
\begin{proof}
Let us consider the loops at Lines \ref{alg2:fl2}-\ref{alg2:fl2end} and Lines \ref{alg2:fl3}-\ref{alg2:fl3end} in sequence.
In the $q^{\mathrm{th}}$ iteration of these loops, host $H_i$ serves block $b_{(i+q-1)\bmod{\beta}}$. From the previous lemmas, after the $(q-1)^{\mathrm{th}}$ iteration of these loops, host $H_i$ has state 
\[\sigma_i = \bigcup_{p=0}^{q-1} b_{(i+p)\bmod{\beta}} \] 
which includes $b_{(i+q-1)\bmod{\beta}}$.
Hence the claim follows.
\end{proof}

\begin{thm}
After the termination of Algorithm \ref{algo2} each host $H_i$, $i \in\{0,...,n-1 \}$, has received all the blocks $b_j \in {\cal B}$ with optimal energy $E(z_2)=\beta ( \Delta_S + \sum_{i=0}^{n-1} \Delta_i )$. 
\end{thm}
\begin{proof}
It follows from Lemma \ref{proof2:lem2} that each host has received all the blocks at the end of the loop at Lines \ref{alg2:fl3}-\ref{alg2:fl3end}. Then, the scheme is correct since each host serves a block that it has already downloaded (Lemma \ref{proof2:lem3}). 
Each host (including the server) is active exactly $\beta$ slots. Then, the total energy consumed is $E(z_2)=\beta ( \Delta_S + \sum_{i=0}^{n-1} \Delta_i )$, which is optimal since it matches the lower bound.
\end{proof}

\subsubsection{Algorithm \ref{algo3}}

For the correctness and optimality proofs of Algorithm \ref{algo3} we define the state $\zeta_{r,\tau}^z$ of a block $b_r$ at the end of $\tau$ as the set of clients $H_i,  i\in\{0,...,n-1\}$, who have received $b_r$. Thus, to start with, $\forall r \in \{0,...,\beta-1\}$, initially the state of block $b_r$ is
$\zeta_{r,0}^z = \emptyset$. After the makespan $\tau_f^z$ of scheme $z$, the state should be, $\forall r\in\{0,...,\beta-1\}$, $\zeta_{r,\tau_f^z}^z = \bigcup_{i=0}^{n-1} \{ H_i\}$

Let us denote the scheme described by Algorithm \ref{algo3} as $z_3$. This scheme has the following properties.

\begin{observation}
After the \emph{for} loop at Lines \ref{alg3:fl1}-\ref{alg3:fl1end},  $\forall r\in \{0,1,..,\beta-1\}$, the state of block $b_r$ is $\zeta_r=\{H_r\}$.
\end{observation}

\begin{lem}
\label{proof3:lem1}
After the $q^{\mathrm{th}}$ iteration of the \emph{for} loop at Lines \ref{alg3:fl2}-\ref{alg3:fl2end}, for $q \in \{0,...,n-\beta\}$, the state of block $b_r$ is
\begin{equation}
\zeta_r = \bigcup_{p=0}^{q} \{H_{r+p}\}
\end{equation}
\end{lem}
\begin{proof}
We prove the claim using induction on $q$. The base case ($q=0$) is trivially true by the observation. Assume the statement to be true for the $(q-1)^{\mathrm{th}}$ iteration. In the $q^{\mathrm{th}}$ iteration, $q=j+1-\beta$. Then, block $b_r$ is served to $H_{r+q}$. Thus, the state of
block $b_r$ after the $q^{\mathrm{th}}$ iteration is
 \[\zeta_r =  \bigcup_{p=0}^{q-1} \{H_{r+p}\} \cup \{H_{r+q}\}= \bigcup_{p=0}^{q} \{H_{r+p}\}
 \]
\end{proof}

\begin{lem}
\label{proof3:lem2}
After the $q'^{\mathrm{th}}$ iteration of the \emph{for} loop at Lines \ref{alg3:fl3}-\ref{alg3:fl3end}, for $q'\in\{0,1,..,\beta-1\}$, the state of block $b_r$ is
\begin{equation}
\zeta_r = \bigcup_{p=0}^{n-\beta} \{ H_{r+p}\} \bigcup_{p=0}^{q'} \{H_{(r-p)\bmod{n}}\}
\end{equation}
\end{lem}

\begin{proof}
The base case ($q'=0$) is true from Lemma \ref{proof3:lem1} after the loop at Lines \ref{alg2:fl2}-\ref{alg2:fl2end} completes.
In iteration $q'=j+1-n$, block $b_{\beta-1}$ is served to $H_{\beta-q'-1}$, hence,
\[\zeta_{\beta-1} = \bigcup_{p=0}^{n-\beta}\{H_{\beta+p-1}\} \bigcup_{p=0}^{q'-1}\{H_{\beta-1-p}\} \cup \{H_{\beta-1-q'}\} \]
and block $b_r, r\in\{0,1,..,\beta-2\}$, is served to $H_{(r-q')\bmod{n}}$. Then, the state of block $b_r$, $r\in\{0,...,\beta-1\}$, after
the $q'^{\mathrm{th}}$ iteration is
\begin{eqnarray*}
\zeta_r & = & \bigcup_{p=0}^{n-\beta} \{ H_{r+p}\} \bigcup_{p=0}^{q'-1} \{H_{(r-p)\bmod{n}}\} \cup \{ H_{(r-q')\bmod{n}}\} \\
& = & \bigcup_{p=0}^{n-\beta} \{ H_{r+p}\} \bigcup_{p=0}^{q'} \{H_{(r-p)\bmod{n}}\}
\end{eqnarray*}
\end{proof}

\begin{lem}
\label{proof3:lem3}
During the execution of Algorithm \ref{algo3}, each host $H_i, i\in \{0,1,..,n-1\}$, serves a block that it has already downloaded.  
\end{lem}
\begin{proof}
In the \emph{for} loop at Lines \ref{alg3:fl2}-\ref{alg3:fl2end}, during iteration $q=j+1-\beta, q\in \{1,..,n-\beta\}$, block $b_r$ is served by $H_{r+q-1}$. It has it because after iteration $q-1$,
\[\zeta_r = \bigcup_{p=0}^{q-1} \{H_{r+p}\},\]
which includes $H_{r+q-1}$. $H_0$ always serves $b_0$, if any, which it has from the above observation. 

In the \emph{for} loop at Lines \ref{alg3:fl3}-\ref{alg3:fl3end}, during iteration $q'=j+1-n, q'\in\{1,..,\beta-1\}$, block $b_{\beta-1}$ is served by $H_{n-q'}$.
It has it because after iteration $q'-1$,
\[\zeta_{\beta-1} = \bigcup_{p=0}^{n-\beta}\{H_{\beta+p-1}\} \bigcup_{p=0}^{q'-1}\{H_{\beta-1-p}\} \cup \{H_{\beta-1-q'}\} \]
which includes $H_{n-q'}, \forall q'\in\{1,2,..,\beta-1\}$.

Block $b_r, r\in \{0,1,..,\beta-2\}$ is served by $H_{(r-(q'-1))\bmod{n}}$.
It has it because after iteration $q'-1$ 
\[\zeta_r = \bigcup_{p=0}^{n-\beta} \{ H_{r+p}\} \bigcup_{p=0}^{q'-1} \{H_{(r-p)\bmod{n}}\}\]
which includes $H_{(r-(q'-1))\bmod{n}}$. Hence, the claim follows.
\end{proof}

\begin{thm}
After the termination of Algorithm \ref{algo3} each host $H_i$, $i \in\{0,...,n-1 \}$ has received all the blocks $b_r \in {\cal B}$ with optimal energy $E(z_3) = \beta \left( \Delta_S + \sum_{i=0}^{n-1} \Delta_i \right) +  (n-\beta) \min\{\Delta_S,\Delta_0\}$.
\end{thm}
\begin{proof}
It follows from Lemma \ref{proof3:lem2} that each host has received all the blocks. Then, the scheme is correct since
each host serves blocks it has already downloaded (Lemma \ref{proof3:lem3}).

We need to bound now the energy consumed. Let us denote $\Delta_{\min}  = \min\{\Delta_S ,\Delta_0 \}$.
The energy consumed in the loop at Lines \ref{alg3:fl1}-\ref{alg3:fl1end} is easily observed to be
\begin{equation}
\label{eq1}
E_1 = \beta\Delta_S + \sum_{i=0}^{\beta-1} \Delta_i 
\end{equation}
The energy consumed in the loop at Lines \ref{alg3:fl2}-\ref{alg3:fl2end} is
\begin{eqnarray}
E_2 & = & \sum_{j=\beta}^{n-1}\left(\Delta_{\min} + \Delta_{j+1-\beta} + \sum_{i=1}^{\beta-1}\Delta_{i+j+1-\beta} \right) \nonumber \\
& = & (n-\beta)\Delta_{\min} + \sum_{j=\beta}^{n-1}\sum_{i=0}^{\beta-1}\Delta_{i+j+1-\beta} \nonumber \\
\label{eq2}
& = & (n-\beta)\Delta_{\min} + \sum_{j=0}^{n-\beta-1}\sum_{i=0}^{\beta-1}\Delta_{i+j+1}
\end{eqnarray}
Finally, the energy consumed in the loop at Lines \ref{alg3:fl3}-\ref{alg3:fl3end} is
\begin{eqnarray}
E_3 & = & \sum_{j=n}^{n+\beta-2} \left( \Delta_{n+\beta-j-2} + \sum_{i=0}^{\beta-2}\Delta_{(n+i-j-1)\bmod{n}} \right) \nonumber \\
& = & \sum_{j=n}^{n+\beta-2}\sum_{i=0}^{\beta-1}\Delta_{(n+i-j-1)\bmod{n}} \nonumber \\
\label{eq3}
& = & \sum_{j=0}^{\beta-2}\sum_{i=0}^{\beta-1}\Delta_{(i-j-1)\bmod{n}}
\end{eqnarray}

Adding Equation \ref{eq1}, \ref{eq2} and \ref{eq3}, we get, 

\begin{eqnarray*}
E(z_3) & = & E_1 + E_2 + E_3 \nonumber \\
& = &  \beta\Delta_S + (n-\beta)\Delta_{\min} + \sum_{i=0}^{\beta-1} \Delta_i \\
& & + \sum_{j=0}^{n-\beta-1}\sum_{i=0}^{\beta-1}\Delta_{i+j+1} + 
 \sum_{j=0}^{\beta-2}\sum_{i=0}^{\beta-1}\Delta_{(i-j-1)\bmod{n}} \nonumber \\
& = &  \beta\Delta_S + (n-\beta)\Delta_{\min} \\
&& + \sum_{i=0}^{\beta-1} \left( \Delta_i  + \sum_{j=i+1}^{i+n-\beta} \Delta_{j} + 
 \sum_{j=0}^{i-1} \Delta_{j} + \sum_{j=i+n-\beta+1}^{n-1} \Delta_{j} \nonumber \right) \\
%& = &  \beta\Delta_S + (n-\beta)\Delta_{\min}  \\
%&& + \sum_{i=0}^{\beta-1}\sum_{j=0}^{\beta-1}\Delta_{(i-j)\bmod{n}} +
%\sum_{i=0}^{\beta-1} \sum_{j=\beta}^{n-1}\Delta_{i+j+1-\beta}   \nonumber \\
& = &  \beta\Delta_S + (n-\beta)\Delta_{\min} + 
\sum_{i=0}^{\beta-1}\sum_{j=0}^{n-1}\Delta_j \nonumber \\
%& =  & \beta\Delta_S + (n-\beta)\Delta_{\min} + \beta\sum_{j=0}^{n-1}\Delta_j \nonumber \\
& =  & \beta \left( \Delta_S + \sum_{j=0}^{n-1}\Delta_j \right) + (n-\beta)\Delta_{\min},
\end{eqnarray*}
which is optimal.
\end{proof}

%-------------------------------------------------------------------------------------

\subsection{Proof of Theorem~\ref{thm-beta}}

From Theorems~\ref{minhetero} and \ref{thmalgo}, the energy consumption of an optimal scheme $z$ in an energy homogeneous system is
\begin{equation}
\label{mink=1}
E(z) = \left( n\beta + \max\{n,\beta\} \right)\cdot\left( \frac{PB}{u\beta} + \delta \right)
\end{equation}
To find the optimal value of $\beta$, we need to minimize the right hand side of Equation \ref{mink=1}. This can be written as a function of $\beta$ as
%Using Equation \ref{mink=1}, we can rewrite Equation \ref{costschblock} as a function of $\beta$ and $\delta$ as,
\begin{numcases}{E(\beta)=}
\frac{PB}{u}(n+1) + \delta (n+1)\beta,&$ \beta \geq n$ \label{first} \\
\frac{nPB}{u}\left(1+\frac{1}{\beta}\right)+\delta n (\beta+1),&$\beta \leq n$ \label{second}
\end{numcases}
Note that in Equation \ref{first} the first term is a constant and the second is linear in $\beta$. This is a straight line with positive slope $\delta (n+1)$.
Hence, the function attains the minimum at the lower extreme $\beta=n$, where it intersects Equation \ref{second}. Hence it is enough to consider
Equation \ref{second} for $\beta \leq n$.
Minimizing Equation \ref{second} with respect to $\beta$ we get, 
\begin{equation}
\label{valuenb-2}
\beta = \sqrt{\frac{PB}{u\delta}}.
\end{equation} 
When this value is larger than $n$ the value $\beta=n$ has to be used.

%--------------------------------------------------------------------------------------
\subsection{Proofs of Theorem \ref{th5}}
\label{sec:th5}

\begin{proof}
It can be easily observed that every slot in which a host receives its first block is a tree slot (since it does not serve anyone). Additionally, no
two clients can receive their first block in the same slot in a normal scheme. Then, 
%, lower bound in Equation \ref{mink=2lb} follows from Lemma \ref{numblocks} as 
there are at least $n$ tree slots.

According to Definition \ref{costblockdef}, the cost $c_{j,i}^z$ of a block can only take values $0$, $\Delta$ or $2\Delta$.
%for a block, either $\mathcal{U}_{j,i}^z+\mathcal{D}_{j,i}^z=0$, $\mathcal{U}_{j,i}^z+\mathcal{D}_{j,i}^z=1$ or $\mathcal{U}_{j,i}^z+\mathcal{D}_{j,i}^z=2$. 
Let us consider a slot $\tau$. We denote with $\#0$, $\#1$, and $\#2$ the number of blocks whose cost is $0$, $\Delta$, and $2\Delta$ in $\tau$, respectively. Then, we can
prove that if $\tau$ is a tree slot, then $\#2=\#0 + 1$, while if $\tau$ is a slot with a cycle, then $\#2=\#0$.
The proof of this claim goes as follows. From Theorem \ref{thmeq}, the cost of all blocks in $\tau$ add up to the cost of $\tau$. Since
all hosts have the same $\Delta$, then $0\cdot\#0 + 1\cdot\#1 + 2\cdot\#2 = |\mathcal{I}_{\tau}^z|$.
In a tree slot the number of blocks served is $\#0 + \#1 + \#2 = |\mathcal{I}_{\tau}^z|-1$, while in a
slot with a cycle the number of blocks served is $\#0 + \#1 + \#2 = |\mathcal{I}_{\tau}^z|$.
Hence the claim follows.

This implies that, if $x$ blocks are served in slot $\tau$, the cost of $\tau$ is $c_{\tau}^z=x \Delta$ if $\tau$ is a slot with a cycle, and $c_{\tau}^z=(x+1) \Delta$ if $\tau$ is a tree slot.
Since the total number of blocks served is $n\beta$ and there are at least $n$ tree slots, the bound follows.
\end{proof}

\subsection{Proofs of Algorithm \ref{algok=2}}
\label{sec:alg4}

The proof of correctness of Algorithm \ref{algok=2} can be divided in essentially four parts. (We use the array abstraction for clarity.) 
The first claim is that, after the first loop (Lines \ref{alg5:fl1}-\ref{alg5:fl1end}), the diagonal of the first subarray has been filled. (I.e., $A_{ii}=1, \forall i\in\{0,...,n-1\}$.) This claim follows trivially by inspection. The second claim is that after the second loop (Lines \ref{alg5:fl2}-\ref{alg5:fl2end}), the top left corner position of each subarray has also been set to 1. (I.e., $A_{0j}=1$,$\forall j\in \{0,n,2n,..,(\lfloor \frac{\beta}{n} \rfloor - 1)n \}$.) This claim also follows by inspection.

The third claim is that, after the $q^{\mathrm{th}}$ iteration of the third loop (Lines \ref{alg5:fl3}-\ref{alg5:fl3end}), the whole $q^{\mathrm{th}}$ subarray and the diagonal of the $(q+1)^{\mathrm{th}}$ subarray have been set to 1 (and the blocks served by a host were available at the host for being served). This can be shown by induction on $q$, where the base case is the first claim above. In the induction step, the proof that
the whole $q^{\mathrm{th}}$ subarray is set to $1$ is similar to the proof of Algorithm \ref{algo1}. The proof that the diagonal of the $(q+1)^{\mathrm{th}}$ subarray is set follows from the second claim above and Line~$\ref{algo5S}$ of the algorithm.

Finally, the fourth claim is that the process described in Line~\ref{run} completes the array. The proof of this claim is very similar to the proof of Algorithm~\ref{algo2}.

Let us now compute the energy consumed by the scheme described by the algorithm. The first loop consumes energy $E_1=2n\Delta$. The second loop consumes
$E_2=2 (\lfloor \beta/n \rfloor - 1) \Delta$. The third loop uses energy
$$
E_3= \Delta \sum_{l=0}^{\lfloor \frac{\beta}{n} \rfloor - 2} \sum_{j=0}^{n-2} (n+1)= \Delta (\lfloor \frac{\beta}{n} \rfloor - 1)(n^2-1)
$$
Finally, the energy consumed by the process described in Line~\ref{run} is
$$
E_4= \Delta \left( \sum_{j=n}^{n+b-1} (n+1) + \sum_{j=n+b}^{n+b+n-2} n \right) = \Delta (b (n+1) +n(n-1)).
$$

Adding up all these terms
$$
E(z_4)= \Delta \left( n(\beta+1) + \left\lfloor \frac{\beta}{n} \right\rfloor  + b - 1\right).
$$

\end{document}